\documentclass{vldb_alt}

\usepackage{verbatim}
\usepackage{suffix}
\usepackage{times}
\usepackage{amsmath,amsfonts}
\usepackage{MnSymbol} 
\usepackage{amssymb}
\usepackage{bbm}
\usepackage{cite}
\usepackage{url}
\usepackage{graphicx}
\usepackage{outlines}
\usepackage{thm-restate}
\usepackage{color}
\usepackage{algpseudocode}
\usepackage{algorithm}
\usepackage{subfigure}
\usepackage{mdwlist}
\usepackage{multirow}
\usepackage{balance}

\begin{document}

\newtheorem{theorem}{Theorem}
\newtheorem{lemma}{Lemma}
\newtheorem{definition}{Definition}
\newtheorem{problem}{Problem}
\newtheorem{proposition}{Proposition}
\newtheorem{corollary}{Corollary}
\newtheorem{example}{Example}

\newcommand{\resp}[1]{\noindent{\color{cyan} #1}}

\newcommand{\gm}[1]{[[\emph{\color{red}GM: #1}]]}

\newcommand{\vect}[1]{\mathbf{#1}}
\newcommand{\E}{\mathbb{E}}

\def\x{\vect{x}}
\def\estx{\vect{\hat x}}
\def\ests{\vect{\hat s}}
\def\W{\vect{W}}
\def\z{\vect{z}}
\def\Z{\vect{Z}}

\def\db{\x}
\def\nbrs{nbrs}
\def\dbdom{\mathbb{Z}_{\geq 0}^n}  
\def\allbuckets{\mathcal{B}}
\def\allhist{\mathcal{P}}
\def\bmean{\frac{b_i(\x)}{|b_i|}}  
\def\unexp{{expand}}

\def\hcost{pcost}
\def\bcost{bcost}

\def\real{\mathbb{R}}
\def\posreal{\mathbb{R}_{\geq 0}}
\def\sens{\Delta}

\newcommand{\Lone}[1]{\left\Vert #1  \right\Vert_1}
\newcommand{\set}[1]{\{#1\}}   

\def\myvert{\;\vert\;}

\def\B{B}  


\def\ones{\vect{1}}
\def\zeros{\vect{0}}
\def\plus{{\!+}}
\def\b{\vect{b}}  
\def\a{\vect{a}}
\def\y{\vect{y}}
\def\q{\vect{q}}  
\def\w{\vect{w}} 
\def\v{\vect{v}}  
\def\estw{\vect{\hat w}}
\def\estq{\vect{\hat q}}
\def\A{\vect{A}}
\def\T{\vect{T}}
\def\Q{\vect{Q}}
\def\M{\vect{M}}
\def\D{\vect{D}}
\def\P{\vect{P}}
\def\p{\vect{p}}
\def\I{\vect{I}}
\def\V{\vect{V}}
\def\H{\vect{H}}
\def\G{\vect{G}}
\def\R{\vect{R}}
\def\X{\vect{X}}
\def\Y{\vect{Y}}
\def\s{\vect{s}}
\def\tq{\hat{q}}
\def\tW{\hat{W}}
\def\tWW{\mathbf{\tW}}
\def\lambdaB{\vect{\lambda}}
\def\LambdaB{\vect{\Lambda}}

\def\esty{\vect{\hat y}}
\def\m{\vect{m}}

\def\dom{\mathbb{N}^n}
\def\rng{\mathcal{Y}}
\def\obj{y}

\def\uni{\mathcal{U}}
\def\coll{\mathcal{S}}
\def\packing{\mathcal{C}}
\def\Lap{\mbox{Laplace}}

\def\dom{\mathcal{D}}
\def\reg{reg}

\def\cover{H}
\def\allparts{\mathcal{P}(\allbuckets)}

\def\DWnospace{\text{DAWA}}
\def\DW{\text{DAWA} }
\def\DWall{\text{DAWA-all} }
\def\DWapprox{\text{DAWA-subset} }
\def\alg{\mathcal{K}}  
\def\LM{\mathcal{L}}	
\def\GM{\mathcal{G}}	
\def\MM{\mathcal{M}}	

\newcommand\debullet[2]{#2}

\title{A Data- and Workload-Aware Algorithm for Range Queries Under Differential Privacy}
\numberofauthors{2}
\author{
	Chao Li$^\dagger$, Michael Hay$^\ddagger$, Gerome Miklau$^\dagger$, Yue Wang$^\dagger$
	\and
	\alignauthor
	       \affaddr{$\dagger$University of Massachusetts Amherst\\}
	       \affaddr{School of Computer Science}\\
	       \affaddr{ \{chaoli,miklau,yuewang\}@cs.umass.edu}
	\alignauthor
	       \affaddr{$\ddagger$Colgate University\\}
	       \affaddr{Department of Computer Science}\\
	       \affaddr{mhay@colgate.edu}
}

\maketitle
\pagestyle{empty}
\begin{abstract}
We describe a new algorithm for answering a given set of range queries under $\epsilon$-differential privacy which often achieves substantially lower error than competing methods.  Our algorithm satisfies differential privacy by adding noise that is adapted to the input data {\em and} to the given query set.  We first privately learn a partitioning of the domain into buckets that suit the input data well.  Then we privately estimate counts for each bucket, doing so in a manner well-suited for the given query set.  Since the performance of the algorithm depends on the input database, we evaluate it on a wide range of real datasets, showing that we can achieve the benefits of data-dependence on both ``easy'' and ``hard'' databases.



\end{abstract}

%
%


\section{Introduction}
\label{sec:intro}

Differential privacy \cite{dwork2006calibrating,dwork2011a-firm} has received growing attention in the research community because it offers both an intuitively appealing and mathematically precise guarantee of privacy. In this paper we study batch (or non-interactive) query answering of range queries under $\epsilon$-differential privacy.  The batch of queries, which we call the {\em workload}, is given as input and the goal of research in this area is to devise differentially private mechanisms that offer the lowest error for any fixed setting of $\epsilon$.  The particular emphasis of this work is to achieve high accuracy for a wide range of possible input databases.

Existing approaches for batch query answering broadly fall into two categories: {\em data-independent} mechanisms and {\em data-dependent} mechanisms.  Data-independent mechanisms achieve the privacy condition by adding noise that is independent of the input database.  The Laplace mechanism is an example of a data-independent mechanism.  Regardless of the input database, the same Laplacian noise distribution is used to answer a query.  More advanced data-indep\-endent mechanisms exploit properties of the workload to achieve greater accuracy, but the noise distribution (and therefore the error) remains fixed for {\em all} input databases.

Data-dependent mechanisms add noise that is customized to properties of the input database, producing different error rates on different input databases.  In some cases, this can result in significantly lower error than data-independent approaches.  These mechanisms typically need to use a portion of the privacy budget to learn about the data or the quality of a current estimate of the data.  They then use the remaining privacy budget to privately answer the desired queries.  In most cases, these approaches do not exploit workload. 

A comparison of state-of-the-art mechanisms in each category reveals that each has advantages, depending on the ``hardness'' of the input database.  If the database is viewed as a histogram, data-bases with large uniform regions can be exploited by these algorithms, allowing the data-dependent mechanisms to outperform data-independent ones.  But on more complex datasets, e.g. those with many regions of density, data-dependent mechanisms break down.

Consider as an example a workload of random range queries and a dataset derived from an IP-level network trace.  A state-of-the-art data-dep\-endent mechanism, {\em Multiplicative Weights and Exponential Mechanism}~(MWEM) \cite{hardt2012a-simple}, can achieve 60.12 average per-query error when $\epsilon=0.1$.  For the same $\epsilon$, one of the best data-independent mechanisms for this workload, {\em Privelet}~\cite{xiao2010differential}, offers per-query error of 196.6, a factor of 3.27 worse.   But other datasets have properties that are difficult to exploit.  On a dataset based on the HEP-PH citation network, MWEM has average per-query error of 722.3 with $\epsilon=0.1$, while the error of the data-independent mechanism is still 196.6 for this workload, a factor of 3.67 better. 

Such a large variation in the relative performance of mechanisms across data sets is a major limitation of current approaches.  This is especially true because it is typically necessary to select a mechanism without seeing the data.
\begin{figure*}[t]
\centering
\subfigure[{\small True database $\x$}]{
\includegraphics[width=.27\textwidth]{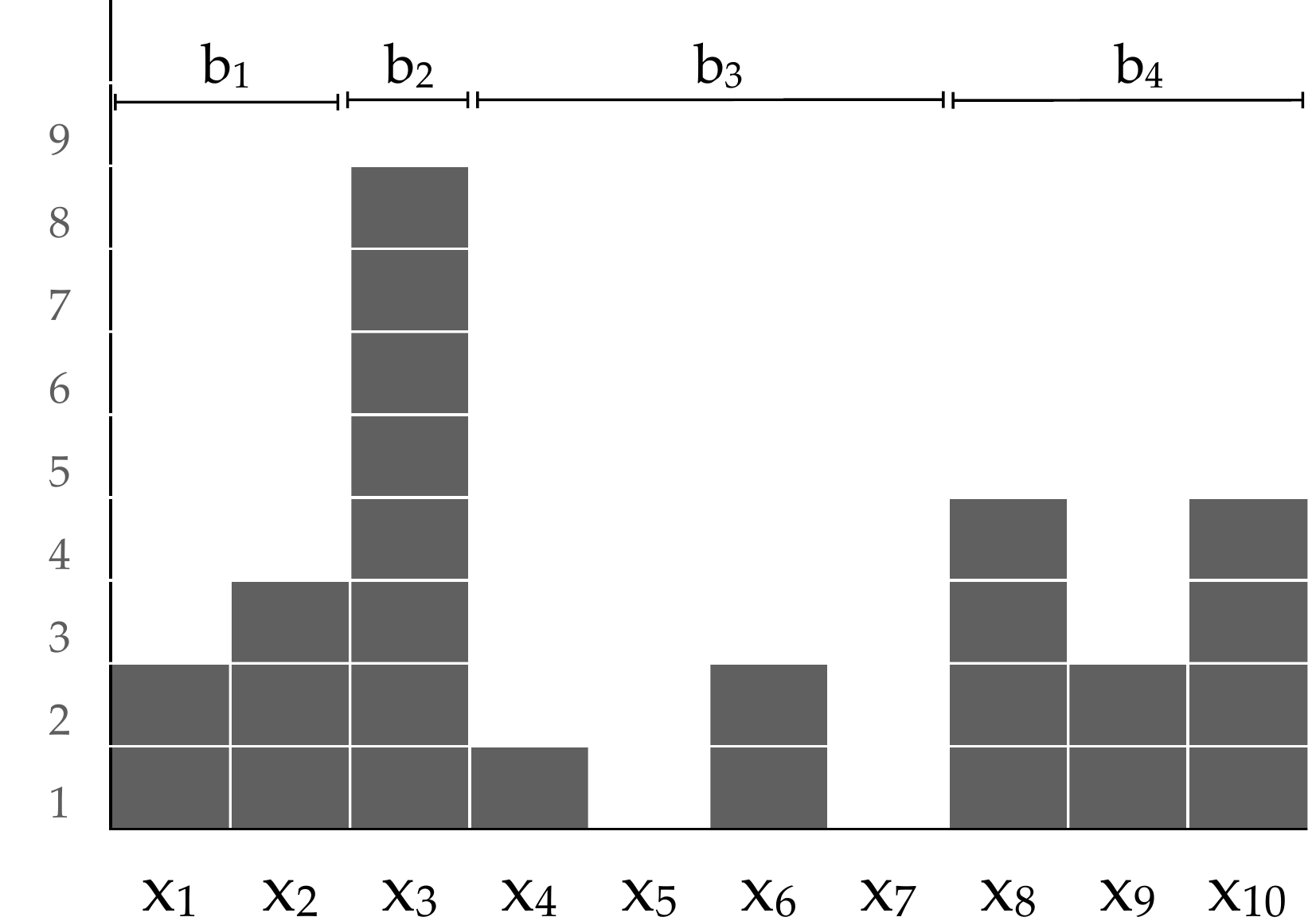}
\quad}
\subfigure[Algorithm flow chart]{
\includegraphics[width=.33\textwidth]{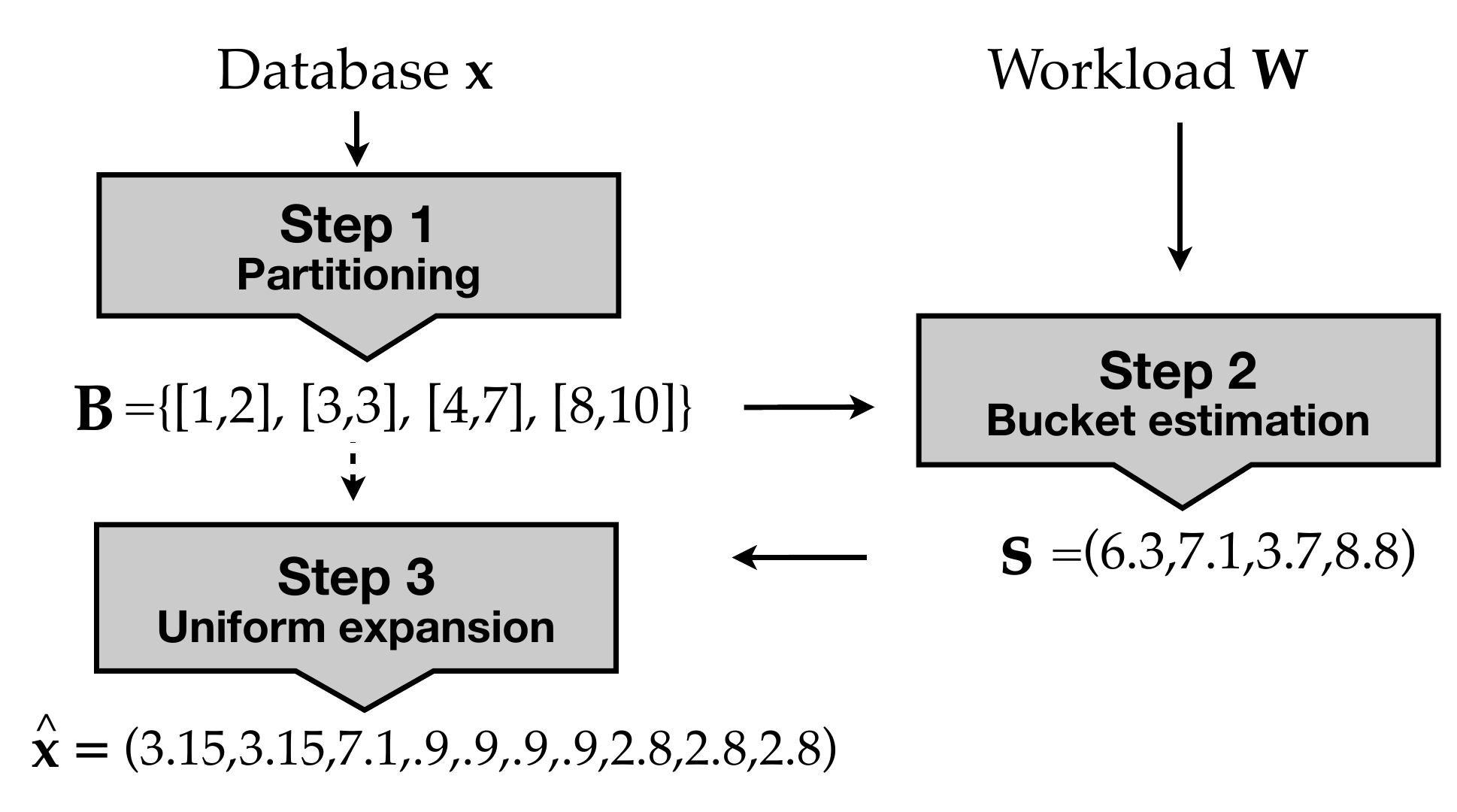}
}
\subfigure[Private output $\estx$]{\quad
\includegraphics[width=.27\textwidth]{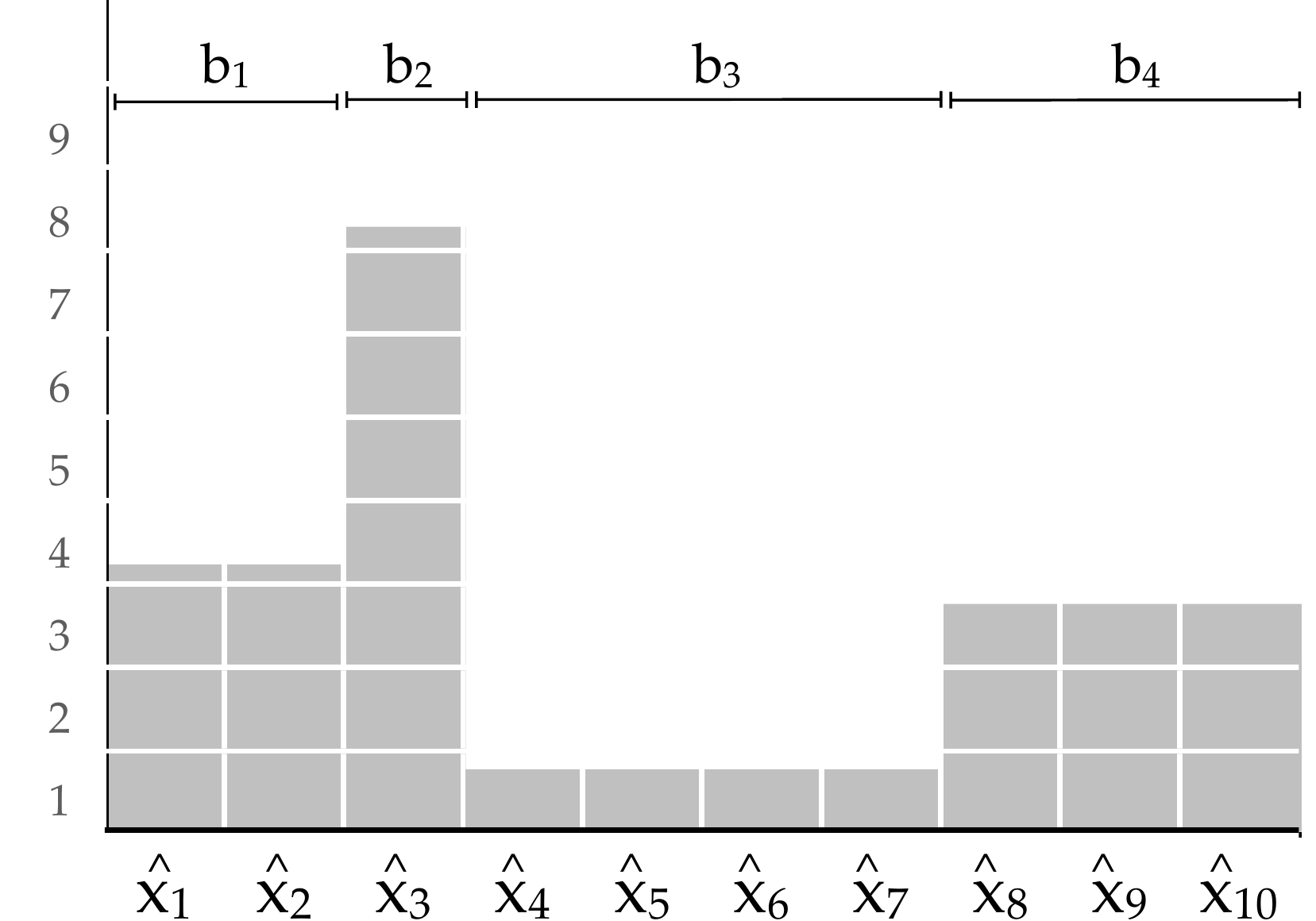}
}
\caption{\label{fig:example} Overview and example execution for the \DW mechanism.}
\end{figure*}

\debullet{\subsubsection*{Contributions}
\begin{itemize} \itemsep -0.01in
\item We propose a novel 2-stage mechanism for answering range queries under $\epsilon$-differential privacy. On inputs where existing data-dependent mechanisms do well, our mechanism achieves lower error by a factor up to $6.86$ compared with the state of the art.  On inputs where existing data-dependent mechanisms do poorly, our mechanism achieves error comparable to state-of-art data-independent mechanisms.

\item We present an efficient algorithm in the first stage that partitions the domain into uniform regions. Compared with other differentially private partitioning algorithms, our algorithm generates much better partitions and runs in time that is only quasilinear in the size of the domain.

\item We design a new, efficient algorithm in the second stage that computes scaling factors for a hierarchical set of range queries. Unlike existing hierarchical strategies, our method allows a non-uniform budget distribution across queries at the same level, which leads to a strategy that is more finely tuned to the workload, and thus more accurate. 
\end{itemize}
}{\paragraph*{Contributions} First, we propose a novel 2-stage mechanism for answering range queries under $\epsilon$-differential privacy. On inputs where existing data-dependent mechanisms do well, our mechanism achieves lower error by a factor of up to $6.86$ compared with the state of the art.  On inputs where existing data-dependent mechanisms do poorly, our mechanism achieves error comparable to state-of-art data-independent mechanisms.  Second, we present an efficient algorithm in the first stage that partitions the domain into uniform regions. Compared with other differentially private partitioning algorithms, our algorithm generates much better partitions and runs in time that is only quasilinear in the size of the domain.  Third, we design a new, efficient algorithm in the second stage that computes scaling factors for a hierarchical set of range queries. Unlike existing hierarchical strategies, our method allows a non-uniform budget distribution across queries at the same level, which leads to a strategy that is more finely tuned to the workload, and thus more accurate.}

To our knowledge, our mechanism is the first data-aware mechanism that provides significant improvement on databases with easy-to-exploit properties yet does not break-down on databases with complex distributions. 

\subsection*{Algorithm Overview}
We give an overview to our new mechanism and an example below.

The {\em Data-Aware/Workload-Aware} (\DWnospace) mechanism is an $\epsilon$-diff\-erentially-private algorithm that takes as input a workload of range queries, $\W$, and a database, $\x$, represented as a vector of counts.  The output is an estimate $\estx$ of $\x$, where  the noise added to achieve privacy is adapted to the input data and to the workload.  
The \DW  algorithm consists of the following three steps, the first two of which require private interactions with the database.  To ensure that the overall algorithm satisfies $\epsilon$-differential privacy, we split the total $\epsilon$ budget into $\epsilon_1$, $\epsilon_2$ such that $\epsilon_1 + \epsilon_2 = \epsilon$ and use these two portions of the budget on the respective stages of the algorithm.
\subsubsection*{Step 1: Private Partitioning}

The first step selects a partition of the domain that fits the input database.  We describe (in Sec.~\ref{sec:partition}) a novel differentially private algorithm that uses $\epsilon_1$ budget to select a partition such that within each partition bucket, the dataset is approximately \emph{uniform}.  This notion of uniformity is later formalized as a cost function but the basic intuition is that if a region is uniform, then there is no benefit in using a limited privacy budget to ask queries at a finer granularity than these regions---the signal is too small to overcome the noise.  The output of this step is $\B$, a partition of $\x$ into $k$ buckets, {\em without} counts for the buckets.

\subsubsection*{Step 2: Private Bucket Count Estimation}

Given the partition $\B$, the second step derives noisy estimates of the bucket counts.  Rather than simply adding Laplace noise to the bucket counts, we use a workload-aware method.  Conceptually, we re-express the workload over the new domain defined by the partition $\B$, with the buckets in the partition taking the place of $\x$.  Then we have a well-studied problem of selecting unbiased measurements (i.e. linear functions of the bucket counts) in a manner that is optimized for the workload.  This problem has received considerable attention in past work \cite{Li:2010Optimizing-Linear,Ding11Differentially,Cormode11Differentially,chaopvldb12,Yuan12Low-Rank,Yaroslavtsev13Accurate}.  We use the basic framework of the matrix mechanism \cite{Li:2010Optimizing-Linear}, but we propose a new algorithm (described in Sec.~\ref{sec:stats}) for efficiently approximating the optimal measurements for the workload.  

Given the selected measurements, we then use the $\epsilon_2$ privacy budget and Laplace noise to privately answer the measurement \\queries, followed by least-squares inference to derive the output of this step, a  noisy estimate $\s$ for the buckets in $\B$.  

\subsubsection*{Step 3: Uniform Expansion}
In the last step we derive an estimate for the $n$ components of $\x$ from the $k$ components of the histogram $(\B,\s)$.  This is done by assuming uniformity: the count $s_i$ for each bucket $b_i$ is spread uniformly amongst each position of $\x$ that is contained in $b_i$.  The result is the estimate $\estx$ for $\x$.  Strictly speaking, any range query can be computed from $\estx$, but the noise is tuned to provide accuracy for precisely the queries in the workload.

The following example illustrates a sample execution of \DW\hspace{-2pt}.

\begin{example}
For $n=10$, Fig.~\ref{fig:example} shows graphically a sample data vector $\x=(2,3,8,1,0,2,0,4,2,4)$.  A possible output of Step 1 is $\B=\set{b_1,b_2,b_3,b_4}$ where $b_1=[1,2]$, $b_2=[3,3]$, $b_3=[4,7]$, and $b_4=[8,10]$. This need not be the optimal partition, as defined in Sec.~\ref{sec:partition}, because the partition selection is randomized. For the sample database $\x$ in the figure, the true bucket counts for the partition would be $(5,8,3,10)$.  
The result from Step 2 is a set of noisy bucket counts, $\s=(6.3, 7.1,3.6,8.4)$.  
Step 3 then constructs $\estx$ by assuming a uniform distribution for values within each bucket.  As it is shown graphically in Fig.~\ref{fig:example}(c), the final output is $$\estx=(3.15,3.15,7.1,.9,.9,.9,.9,2.8,2.8,2.8).$$
\end{example}

%
%
%

The novelty of our approach consists of splitting the overall private estimation problem into two phases: Step 1, which is data-dependent, and Step 2, which is workload-aware.  Our main technical contributions are an effective and efficient private solution to the optimization problem underlying Step 1, and an effective and efficient solution to the optimization problem underlying Step 2.  We also extend our methods to two-dimensional workloads using spatial decomposition techniques.

A number of recently-proposed methods \cite{Acs2012compression,Xiao:2012fk,Cormode11Differentially,xu2013differential} share commonalities with one or more parts of our mechanism (as described in Sec.~\ref{sec:related}).  But each omits or simplifies an important step and/or they use sub-optimal methods for solving related subproblems.  In Sec.~\ref{sec:experiments}, an extensive experimental evaluation shows that for workloads of 1- and 2-dimensional range queries, the \DW algorithm achieves lower error than all competitors on nearly every database and setting of $\epsilon$ tested, often by a significant margin.  

The paper is organized as follows. We review notation and privacy definitions in Sec.~\ref{sec:background}. The partitioning algorithm is presented in Sec.~\ref{sec:partition}, and the bucket count estimating algorithm is included in Sec.~\ref{sec:stats}. We extensively compare \DW with state-of-the-art competing mechanisms in Sec.~\ref{sec:experiments}. Related work is discussed in Sec.~\ref{sec:related}. We conclude and mention future directions in Sec.~\ref{sec:conclusion}.

%

\section{Background}\label{sec:background}

In this section we review notation, basic privacy definitions, and standard privacy mechanisms used throughout the paper.

\subsection{Databases and Queries}

The query workloads we consider consist of counting queries over a single relation.  Let the database $I$ be an instance of a single-relation schema $R(\mathbb{A})$, with attributes $\mathbb{A}=\{A_1, A_2, \ldots, A_k\}$ each having an ordered domain.  In order to express our queries, we first transform the instance $I$ into a {\em data vector} $\x$ consisting of $n$ non-negative integral counts.  We restrict our attention to the one- or two-dimensional case.  In one dimension, we isolate a single attribute, $A_i$, and define $\x$ to consist of one coefficient for each element in the domain, $dom(A_i)$.  In other words, $x_j$ reports the number of tuples in database instance $I$ that take on the $j^{th}$ value in the ordered domain of $A_i$.  In the two-dimensional case, for attributes $A_i,A_j$, $\x$ contains a count for each element in $dom(A_i)\times dom(A_j)$.  For simplicity, we describe our methods in the one-dimensional case, extending to two-dimensions in Sec. \ref{sec:exptspatial}.

A query workload $\W$ defined on $\x$ is a set of range queries $\set{w_1 \dots w_m}$ where each $w_i$ is described by an interval $[j_1,j_2]$ for $1 \leq j_1 \leq j_2 \leq n$.  The evaluation of $w_i=[j_1, j_2]$ on $\x$ is written $w_i(\x)$ and defined as $\sum_{j=j_1}^{j_2}{x_j}$.  We use $\W(\x)$ to denote the vector of all workload query answers $\langle w_1(\x) \dots w_m(\x) \rangle$.

A histogram on $\x$ is a partition of $[1,n]$ into non-overlapping intervals, called buckets, along with a summary statistic for each bucket.  We denote a histogram by $(\B,\s)$ with $\B$ a set of buckets $\B=\set{b_1 \dots b_k}$ and $\s$ a set of corresponding statistics $\s = s_1 \dots s_k$.  Each $b_i$ is described by an interval $[j_1,j_2]$ and the set of intervals covers $[1,n]$ and all intervals are disjoint.  We define the length $|b_i|$ of bucket $b_i$ to be $j_2-j_1+1$.

We associate a summary statistic with each of the $k$ buckets in a histogram. One way to do this is to treat the bucket intervals as range queries and evaluate them on $\x$.  We denote this true statistic for bucket $b_i$ by $b_i(\x)$ and we use $\B(\x)$ to denote the vector for true bucket counts.  In other cases, the summary statistics are noisy estimates of $\B(\x)$, denoted $\s = {s}_1 \dots {s}_k$.

Throughout the paper we use the {\em uniform expansion} of a histogram $(\B, \s)$.  It is a data vector of length $n$ derived from $\B$ by assuming uniformity for counts that fall within bucket ranges.  
\begin{definition}[Uniform Expansion] Let $\unexp$ be a\\ function that takes a histogram $H=(\B,\s)$ with buckets $\B=\set{b_1 \dots b_k}$ and statistics $\s=s_1 \dots s_k$, and uniformly expands it.  Thus, $\unexp(\B, \s)$ is an $n$-length vector $\y$ defined as:
$$y_j = \frac{s_{t(j)} }{ |b_{t(j)}|}$$
where $t(j)$ is the function that maps position $j$ to the index of the unique bucket in $\B$ that contains position $j$ for $j \in [1,n]$.
\end{definition}
In our algorithms, both the choice of a histogram and the value of the histogram statistics have the potential to leak sensitive information about $\x$.  Both must be computed by a differentially private algorithm.  Suppose that a differentially private algorithm returns histogram $H = (\B, \s)$ where the statistics have noise added for privacy.  We use $\estx$ to denote the uniform expansion of $H$, i.e.,  $\estx = \unexp(\B, \s)$.  Since the vector $\estx$ is a differentially private estimate for $\x$, we can use it to answer any query $w$ as $w(\estx)$.

We are interested in how accurately $\estx$ approximates $\x$.  The \emph{absolute error} of $\estx$ is defined as $\Lone{\x - \estx }$.  
The \emph{expected absolute error} is $\E\Lone{\x - \estx }$ where the expectation is taken over the randomness of $\estx$.
Given workload $\W$, the \emph{average error} on $\W$ is $\frac{1}{m} \Lone{\W(\x) - \W(\estx)}$.


\subsection{Private Mechanisms}

Differential privacy places a bound (controlled by $\epsilon$) on the difference in the probability of algorithm outputs for any two {\em neighboring} databases.  For database instance $I$, let $\nbrs(I)$ denote the set of databases differing from $I$ in at most one record; i.e., if $I' \in \nbrs(I)$, then $|(I - I') \cup (I' - I)| = 1$.

\begin{definition}[Differential Privacy~\cite{dwork2006calibrating}] A randomized algorithm $\alg$ is $\epsilon$-differentially private if for any instance $I$, any $I' \in \nbrs(I)$, and any subset of outputs $S \subseteq Range(\alg)$, the following holds:
\[
Pr[ \alg(I) \in S] \leq \exp(\epsilon) \times Pr[ \alg(I') \in S]
\]		
\end{definition}
Differential privacy has two important composition properties~\cite{mcsherry2009pinq}.  Consider $k$ algorithms $\alg_1, \dots, \alg_k$, each satisfying $\epsilon_i$-differential privacy.  The \emph{sequential} execution of $\alg_1, \dots, \alg_k$ satisfies $(\sum \epsilon_i)$-differential privacy.  Suppose the domain is partitioned into $k$ arbitrary disjoint subsets and $\alg_i$ is executed on the subset of data from the $i^{th}$ partition.  The \emph{parallel} execution of $\alg_1, \dots, \alg_k$ satisfies $(\max_i \set{\epsilon_i})$-differential privacy.

For functions that produce numerical outputs, differential privacy can be satisfied by adding appropriately scaled random noise to the output.  The scale of the noise depends on the function's \emph{sensitivity}, which captures the maximum difference in answers between any two neighboring databases. 
\begin{definition}[Sensitivity]
    Given function $f:$\\ $dom(A_1) \times \dots \times dom(A_k) \rightarrow \real^d$, the sensitivity of $f$, denoted $\sens f$, is defined as: 
\[
\sens f = \max_{I, I' \in \nbrs(I)} \Lone{ f(I) - f(I') }
\]
\end{definition}
Sensitivity extends naturally to a function $g$ that operates on data vector $\x$ by simply considering the composition of $g$ with the function that transforms instance $I$ to vector $\x$.  In this paper, we consider functions that take additional inputs from some public domain $\mathcal{R}$.  For such functions, $\sens f$ measures the largest change over all pairs of neighboring databases and all $r \in \mathcal{R}$. 

The Laplace mechanism achieves differential privacy by adding Laplace noise to a function's output.  We use $\Lap(\sigma)$ to denote the Laplace probability distribution with mean 0 and scale $\sigma$.
\begin{definition}[Laplace Mechanism~\cite{dwork2006calibrating}]
Given function $f: dom(A_1) \times \dots \times dom(A_k) \rightarrow \real^d$, let $\z$ be a $d$-length vector of random variables where $z_i \sim \Lap(\sens f / \epsilon)$.  The Laplace mechanism $\LM$ is defined as $\LM(I) = f(I) +  \z$.
\end{definition}

\section{Private Partitioning} \label{sec:partition}

This section describes the first stage of the \DW algorithm.  The output of this stage is a partition $\B$.  In Sec.~\ref{sec:partitioncost}, we motivate the problem of finding a good partition and argue that the quality of a partition depends on the data.  We then describe a differentially private algorithm for finding a good partition in Sec.~\ref{sec:privatepartition}.

This stage of \DW is not tuned to the workload of queries and instead tries to select buckets such that, after statistics have been computed for the buckets and the histogram is uniformly expanded, the resulting $\estx$ is as close to $\x$ as possible.

\subsection{Cost of a partition} \label{sec:partitioncost}

Recall that after the partition $\B=\set{b_1 \dots b_k}$ has been selected, corresponding statistics $s_1, \dots, s_k$ are computed.  Let ${s}_i = b_i(\x) + Z_i$ where $Z_i$ is a random variable representing the noise added to ensure privacy.  (This noise is added in the second stage of \DWnospace.)  Once computed, the statistics are uniformly expanded into $\estx = \unexp(\B, s)$, which is an estimate for $\x$.  If bucket $b_i$ spans the interval $[j_1, j_2]$ we use $j \in b_i$ to denote $j \in [j_1, j_2]$. After applying uniform expansion, the resulting estimate for $x_j$, for $j \in b_i$, is:
\begin{align}
\hat{x}_j &= \bmean + \frac{Z_i}{|b_i|}  \label{eqn:xi} 
\end{align}
The accuracy of the estimate depends on two factors.  The first factor is the bucket size.  Since the scale of $Z_i$ is fixed, larger buckets have less noise per individual $\hat{x}_j$.  The second factor is the degree of uniformity within the bucket.  Uniform buckets, where each $x_j$ is near the mean of the bucket $\bmean$, yield more accurate estimates.  


We can translate these observations about $\hat{x}_j$ into a bound on the expected error of $\estx$.  For bucket $b_i$, let $dev$ be a function that measures the amount the bucket \emph{deviates} from being perfectly uniform: 
\begin{equation}
dev(\x, b_i) = \sum_{j \in b_i} \left|x_j - \bmean\right|    
\end{equation}
The bound on the expected error of $\estx$ is in terms of the deviation and the error due to added noise.  
\begin{proposition}
\label{prop:histerror}
Given histogram $H=(\B, \s)$ where $|\B| = k$ and for $i=1 \dots k$, ${s}_i = b_i(\x) + Z_i$ where $Z_i$ is a random variable.  The uniform expansion, $\estx = \unexp(\B, \s)$, has expected error
\begin{equation} \label{eqn:experror}
    \E \Lone{\estx - \x} \leq \sum_{i=1}^k dev(\x, b_i) + \sum_{i=1}^k \E |Z_i|
\end{equation}
\end{proposition}
The proof of this bound follows from (\ref{eqn:xi}) and the fact that $|a + b| \leq |a| + |b|$.  Proof of a similar result is given in Acs et al.~\cite{Acs2012compression}.

Prop.~\ref{prop:histerror} reveals that the expected error of a histogram can be decomposed into two components: (a) \emph{approximation error} due to approximating each $x_j$ in the interval by the mean value $\bmean$ and (b) \emph{perturbation error} due to the addition of random noise.  The perturbation component is in terms of random variables $Z_i$, which are not fully determined until the second stage of \DWnospace.  For the moment, let us make the simplifying assumption that the second stage uses the Laplace mechanism (with a budget of $\epsilon_2$).  Under this assumption, $Z_i \sim \Lap(1/\epsilon_2)$ and 
$\sum_{i=1}^k \E |Z_i|$ simplifies to $k/\epsilon_2$.  This error bound conforms with our earlier intuition that we want a histogram with fewer (and therefore larger) buckets that are as uniform as possible.  The optimal choice depends on the uniformity of the dataset $\x$ and on the budget allocated to the second stage (because smaller $\epsilon_2$ increases perturbation error, making less uniform buckets relatively more tolerable).


We use Prop.~\ref{prop:histerror} as the basis for a cost function.
\begin{definition}[Cost of partition] \label{def:l1cost}
Given a partition of the domain into buckets $\B = \set{b_1, \dots, b_k}$, the cost of $\B$ is
\begin{equation}
    \hcost(\x, \B) = \sum_{i=1}^k dev(\x, b_i) \;+\; k / \epsilon_2 
\end{equation}
\end{definition}
This cost function is based on the simplifying assumption that $Z_i \sim \Lap(1/\epsilon_2)$.  In fact, in the \DW algorithm, each $Z_i$ is a weight\-ed combination of Laplace random variables.  The weights, which are tuned to the workload, are not selected until the second stage of \DWnospace.  However, any weight selection has the property that $\E |Z_i| \geq 1/\epsilon_2$.  This means our choice of cost function is conservative in the sense that it favors a more fine-grained partition than would be selected with full knowledge of the noise distribution. 

\begin{example}
Recall the partition $\B=\set{b_1,b_2,b_3,b_4}$ in Fig.~\ref{fig:example}.

\begin{itemize}\itemsep 0in
\item $b_1=[1,2]$, $\frac{b_1(\x)}{|b_1|}=\frac{5}{2}$, $dev(\x,b_1)=\frac{1}{2} + \frac{1}{2} = 1$
\item $b_2=[3,3]$, $\frac{b_2(\x)}{|b_2|}=\frac{8}{1}$, $dev(\x,b_2)=0$
\item $b_3=[4,7]$, $\frac{b_3(\x)}{|b_3|}=\frac{3}{4}$, $dev(\x,b_3)=\frac{1}{4}+\frac{3}{4}+\frac{5}{4}+\frac{3}{4}=3$
\item $b_4=[8,10]$, $\frac{b_4(\x)}{|b_4|}=\frac{10}{3}$, $dev(\x,b_4)= \frac{2}{3} + \frac{4}{3} + \frac{2}{3}=2\frac{2}{3} $
\end{itemize}
Therefore, $pcost(\x,B) = 6\frac{2}{3} + 4/\epsilon_2$.
When $\epsilon_2=1.0$, $pcost(\x,B)= 6\frac{2}{3} + 4 = 10\frac{2}{3}$.
In comparison, the cost of partitioning $\x$ as a single bucket $[1,10]$ leads to a deviation of $17.2$ and total $pcost$ of $18.2$.  Thus $B$ is a lower cost partition and intuitively it captures the structure of $\x$ which has four regions of roughly uniform density.  But note that with a more stringent privacy budget of $\epsilon_2=0.1$, the perturbation error per bucket rises so $pcost(\x, B) = 6\frac{2}{3} + 40 = 46\frac{2}{3}$ whereas the $pcost$ of a single bucket is only $17.2 + 10 = 27.2$.
\end{example} 

Given this cost function, we can now formally state the problem that the first stage of \DW aims to solve.
\begin{problem}[Least Cost Partition Problem] \label{problem:optpartition} 
The least cost partition problem is to find the partition that minimizes the following objective:
    \begin{equation*} 
     \begin{aligned}
     & \underset{\B \subseteq \allbuckets}{\text{minimize}}
     & &  \hcost(\db, \B) \\
     & \text{subject to}
     & & 
      \bigcup_{b \in \B} b = [1,n], \text{ and }
      \forall \; b,b' \in \B, b \cap b' = \emptyset
     \end{aligned}
    \end{equation*}
where $\allbuckets$ is the set of all possible intervals $\allbuckets = \set{ [i, j] \myvert 1 \leq i \leq j \leq n }$ and the constraint ensures that $\B$ partitions [1,n].  
\end{problem}
The next section describes our algorithm for solving this optimization problem in a differentially private manner.

\subsection{Finding a least cost partition} \label{sec:privatepartition}

Since partition cost is data-dependent, we cannot solve Problem~\ref{problem:optpartition} exactly without violating privacy.  Instead, we must introduce sufficient randomness to ensure differential privacy.  
%
%
Our approach is efficient and simple; our main contribution is in showing that this simple approach is in fact differentially private.  

Our approach is based on the observation that the cost of a partition decomposes into a cost per bucket.  Let $\bcost$ be a function that measures the cost of an individual bucket $b$,
    $$\bcost(\x, b) = dev(\x, b) + 1/\epsilon_2.$$
For any partition $\B$, the partition cost is simply the sum of the bucket costs: $\hcost(\x, \B) = \sum_{b \in \B} \bcost(\x, b)$.  
Since one needs to interact with the private database in computing the cost of each bucket, reporting the partition with the least cost will violate differential privacy. Instead, we solve Problem~\ref{problem:optpartition} using \emph{noisy} cost: the noisy cost of a bucket comes from perturbing its bucket cost with a random variable sampled from the Laplace distribution, and the noisy partition cost is the sum of the noisy bucket costs.
    
The algorithm for this stage is shown in Algorithm~\ref{alg:privpartintervals}.  It takes as input the private database $\x$ as well as $\epsilon_1$ and $\epsilon_2$.  The parameter $\epsilon_1$ represents the privacy budget allocated to this stage.  The parameter $\epsilon_2$ represents the privacy budget allocated to the second stage (Algorithm~\ref{alg:greedyhier}, Sec.~\ref{sec:stats}).  That parameter is needed here because the value of $\epsilon_2$ is used in calculating the bucket costs.

Algorithm~\ref{alg:privpartintervals} has three simple steps.  First, it calls the subroutine {\sc AllCosts} to efficiently compute the cost for all possible buckets (details are below).  
Second, it adds noise to each bucket cost.  Finally, it calls the {\sc LeastCostPartition} subroutine to find the partition with the least \emph{noisy} cost.  This is done using dynamic programming, much like classical algorithms for v-optimal histograms~\cite{Jagadish:1998:OHQ:645924.671191}.  

\begin{algorithm}
\small
\caption{Private partition for intervals and $L_1$ cost function}
\label{alg:privpartintervals}
\begin{algorithmic}
    \Procedure{Private Partition}{$\db$, $\epsilon_1$, $\epsilon_2$}
    \State // Let $\allbuckets$ be the set of all intervals on $[1,n]$
    \State // Compute cost $\bcost(\x, b)$ for all $b \in \allbuckets$
    \State $cost \gets \Call{AllCosts}{\db, \epsilon_2}$ 
    \State // Add noise to each bucket cost
    \For {$b \in \allbuckets$}
        \State $cost[b] \gets cost[b] + Z$, where $Z \sim \Lap(2 \sens \bcost/\epsilon_1)$ 
    \EndFor
    \State // Find $\B$ with lowest total cost based on noisy bucket costs
    \State // stored in $cost$
    \State $\B \gets \Call{LeastCostPartition}{\allbuckets, cost}$
    \State \Return $\B$
    \EndProcedure
\end{algorithmic}
\end{algorithm}

We analyze Algorithm~\ref{alg:privpartintervals} along three key dimensions: accuracy, computational efficiency, and privacy.  

\paragraph*{Accuracy} Accuracy is measured in terms of the difference in cost between the selected partition and the optimal choice (ignoring privacy).
We give the following bound on the algorithm's accuracy.  
\begin{restatable}{theorem}{thmutilpartition}
    \label{thm:utility_partition}
With probability at least $1 - \delta$, Algorithm~\ref{alg:privpartintervals} returns a solution with cost at most $OPT + t$ where $OPT$ is the cost of the least cost solution and $t = 4 \sens c \; n \log (|\allbuckets|/\delta) / \epsilon_1$.
\end{restatable}
In addition to a theoretical analysis, we do an extensive empirical evaluation in Sec.~\ref{sec:experiments}.

%

\paragraph*{Efficiency} 


The computationally challenging part is {\sc AllCosts}, which computes the cost for each bucket.   
%
Unlike the bucket cost for a v-optimal histogram (which is based on an $L_2$ metric, rather than the $L_1$ metric used here), the cost does not decompose easily into sum and sum of square terms that can be precomputed.  Nevertheless, we show that we can decompose the the cost into partial sums of $x_j$ which can be computed using a balanced tree.

Given bucket $b_i$, 
let us identify the indexes $j \in b_i$ that are above the bucket mean, $\bmean$.  Let $I^+ = \left\{ j \myvert j \in b_i \text{ and } x_j\geq \bmean \right\}$.  Let $I^{-}$ be those below the mean, $I^{-} = b_i - I^{+}$.  We can simplify $dev(\x, b_i)$ as follows:
\begin{align*}
dev(\x, b_i)
&=\sum_{j\in I^{+}}\left(x_j - \bmean \right)+\sum_{j\in I^{-}}\left(\bmean  - x_j \right)\\
&=2\sum_{j\in I^{+}}\left(x_j - \bmean \right)\\
&=2 \sum_{j\in I^{+}}x_j -|I^{+}|\cdot \bmean
\end{align*}
The second equality follows from the fact that the sum of deviations above the mean must be equal to the sum of deviations below the mean.
The above equation implies that the total deviation can be computed knowing only the sum of $x_j$ for $j \in I^{+}$ and the size of $I^{+}$.  Those quantities can be efficiently computed using a binary search tree of $x_{j_1}, \ldots, x_{j_2}$. Each node in the tree stores a value (some $x_j$) as well as the sum of all values in its subtree, and the number of nodes in its subtree.  For any constant $a$, we can then compute $\sum_{j \in b_i, x_j \geq a} (x_j - a)$ via binary search.  

To compute the bucket costs for all intervals with length $\ell$, we can dynamically update the search tree. After the cost for interval $[j,j+\ell]$ has been computed, we can update the tree to compute interval $[j+1, j+\ell+1]$ by removing $x_j$ from the tree and adding $x_{j+\ell+1}$.  Using a self-balancing tree, computing all intervals of size $\ell$ requires $O(n \log n)$ time.  To compute \emph{all} intervals, the total runtime is $O(n^2 \log n)$. 

We can reduce the runtime to $O(n \log^2 n)$ by restricting to intervals whose length is a power of two.
This restriction has the potential to exclude the optimal solution.  Empirically, we find that Algorithm~\ref{alg:privpartintervals} remains almost as accurate as when it uses all intervals, and is always more accurate than competing techniques (Sec.~\ref{sec:exptpartitions}).  The benefit of the approximation is reduced runtime, which makes it feasible to run on larger datasets.

The last step of Algorithm~\ref{alg:privpartintervals}, {\sc LeastCostPartition}, is efficient, requiring time linear in $n$ and the number of buckets.


%
%

\paragraph*{Privacy}
The proof of privacy is the main challenge.  Analyzing the privacy requires some subtlety because the noise by itself is not necessarily enough to guarantee privacy.  (If the Laplace mechanism was used to publish noisy costs for {\em all} buckets, the scale of the noise would be $\Omega(n)$.)  However, when the actual noisy costs are kept secret and the only published output is the partition with the least (noisy) cost, then a small amount of noise is sufficient to ensure privacy.  The noise is proportional to the sensitivity of the bucket cost. It can be shown $\sens \bcost \leq 2$. 
\begin{theorem}
\label{thm:privacy}
Algorithm~\ref{alg:privpartintervals} is $\epsilon_1$-differentially private.
\end{theorem}



\newcommand{\zbnd}[1]{Z(#1, \z^{-i})}   
\newcommand{\indic}[1]{\textbf{I} \left[ #1 \right]}
\def\alg{\mathcal{A}}

\begin{proof}[of Theorem~\ref{thm:privacy}]
Recall that $\allbuckets = \set{ [i, j] \myvert 1 \leq i \leq j \leq n }$ is the set of all intervals.   For convenience, we make a few small adjustments to notation.  First, we index this set: let $\allbuckets = \set{b_1, \dots, b_M}$ where $M = |\allbuckets|$.  Second, we describe a partition $B$ in terms of this indexed set, so we say $B = \set{i_1, \dots, i_k }$ to mean that $B$ consists of the intervals $b_{i_1}, \dots, b_{i_k}$.  A partition $B$ is valid if it covers the domain and its buckets are disjoint.  Let $\allhist$ be the set of all valid partitions.  Finally, we use this same indexing for the random variables that represent the added noise:  let $\Z = (Z_1, \dots, Z_M)$  where for each $i \in [1,M]$, the random variable $Z_i \sim \Lap(\lambda)$ represents the noise added to the cost of $b_i$.

Let $\db_0,\db_1$ be any pair of neighboring databases and let $B \in \allhist$ be any output of the algorithm.   It suffices to prove
\begin{align*}
P(\alg(\db_1) = B) \geq e^{-\epsilon} P(\alg(\db_0) = B) 
\end{align*}
where $\alg(\db)$ denotes Algorithm~\ref{alg:privpartintervals} running on input $\db$ and the probability distribution is over random variables $\Z$.

When run on input $\db$, the algorithm will output partition $B$ if and only if $B$ is the partition with the lowest noisy cost.  Formally, let $\z \in \real^M$ represent an assignment of $\Z$.  Partition $B$ will be selected if and only if $\z$ satisfies the following condition:
$$
\sum_{j \in B} \bcost(\db, b_j) + z_j < \min_{B' \in \allhist - \set{B}} \left\{ \sum_{k \in B'} \bcost(\db, b_k) + z_k \right\}
$$

Since neighboring databases $\db_0$ and $\db_1$ only differ by one record and the buckets of $B$ partition the domain, there must be exactly one $i\!\in\!B$ where $\bcost(\db_0,b_i)\!\neq\!\bcost(\db_1,b_i)$.  We will now derive an expression for the probability that $B$ is selected that focuses on the noisy cost of bucket $b_i$.  To do this, it will be convenient to partition the space of possible partitions into those that include bucket $b_i$ and those that do not.  Let $\allhist^+ = \set{ B \myvert B \in \allhist \text{ and } i \in B }$ and let $\allhist^- = \allhist - \allhist^+$.  $B$ will be selected if and only if (a) $B$ is the partition with least noisy cost in $\allhist^+$ and (b) $B$ has lower noisy cost than any partition in $\allhist^-$.  We examine these two conditions in turn.

For condition (a), observe that all partitions in $\allhist^+$ use bucket $b_i$, thus whether (a) holds is independent of the outcome of $Z_i$ since it has the same effect on the scores of all partitions in $\allhist^+$.  We use $\z^{-i}$ as shorthand for $(z_1, \dots, z_{i-1}, z_{i+1}, \dots, z_n)$.  Let $\phi$ be a predicate that is true if and only if the assignment of $\z^{-i}$ makes $B$ the least cost partition among $\allhist^+$, and false otherwise:  
\begin{align*}
&\phi({\db}, \z^{-i}) \\    
&= \!\!\!\sum_{j \in B - \set{i}}\!\!\! \bcost(\db, b_j) + z_j < \!\!\min_{B' \in \allhist^+ - \set{B}} \left\{ \sum_{k \in B' - \set{i}} \!\!\!\bcost(\db, b_k) + z_k \right\}
\end{align*}
Since $\db_0$ and $\db_1$ only differ in the score assigned to bucket $b_i$, $\phi(\db_0, \z^{-i}) = \phi(\db_1, \z^{-i})$ 
for all $\z^{-i} \in \real^{M-1}$.

For condition (b), let $\psi$ be a predicate that is true if and only if the assignment of $\z$ makes $B$ a lower cost partition than any partition in $\allhist^-$, and false otherwise.  A key insight is that if we fix $\z^{-i}$, then $B$ will have lower cost provided that $z_i$ is small enough.
\begin{align*}
\psi(\db, \z) &= \sum_{j \in B} \bcost(\db, b_j) + z_j < \min_{B' \in \allhist^-} \left\{ \sum_{k \in B'} \bcost(\db, b_k) + z_k \right\} \\
&= z_i < \zbnd{\db}
\end{align*}
where 
\begin{align*}
&\zbnd{\db} = \\    
& \min_{B' \in \allhist^-} \left\{ \sum_{k \in B'} \bcost(\db, b_k) + z_k \right\} - \sum_{j \in B} \bcost(\db, b_j) - \sum_{\ell \in B - \set{i}} z_\ell
\end{align*}
The upper bound $\zbnd{\db}$ depends on the database.  For neighboring databases $\db_0$ and $\db_1$,
$\zbnd{\db_1} \geq \zbnd{\db_0} - 2 \sens \bcost$. This is
because compared to the score on $\db_0$, the score on neighboring database $\db_1$ of the minimum cost partition in $\allhist^{-}$ could be lower by at most $\sens \bcost$ and the cost of $B$ could be larger by at most $\sens \bcost$.  

We can now express the probability that the algorithm on input $\db$ outputs $B$ in terms of $\phi$ and $\psi$.  Let $f_{\Z}$ (respectively $f_Z$) denote the density function for a multivariate (respectively univariate) Laplace random variable, and $\indic{ \cdot }$ denote the indicator function.
\begin{align*}
&P(\alg(\db) = B) = P( \phi(\db, \Z^{-i}) \land \psi(\db, \Z) ) \\
&= \int \indic{ \phi(\db, \z^{-i}) \land \psi(\db, \z) } f_{\Z}(\z) \mathrm{d}\z \\
&= \int \indic{ \phi(\db, \z^{-i}) } f_{\Z^{-i}}(\z^{-i}) \left( \int  \indic{  \psi(\db, \z) } f_{Z_i}(z_i) \mathrm{d}z_i \right) \mathrm{d}\z^{-i} \\
&= \int \indic{ \phi(\db, \z^{-i}) } f_{\Z^{-i}}(\z^{-i}) P(Z_i < \zbnd{\db}) \mathrm{d}\z^{-i}
\end{align*}
Since $P(Z_i < C)$ decreases with decreasing $C$, we have for neighboring databases $\db_0$ and $\db_1$ and any $\z^{-i} \in \real^{M-1}$ that
\begin{align*}
&P(Z_i < \zbnd{\db_1}) \\
&\geq P(Z_i < \zbnd{\db_0} - 2\sens \bcost) \\
&\geq e^{- 2\sens \bcost/\lambda} P(Z_i < \zbnd{\db_0} ) 
\end{align*}
where the last line follows from the fact that if $Z$ is a Laplace random variable with scale $\lambda$, then for any $z$ and any constant $c > 0$, $P(Z < z - c) \geq e^{-c/\lambda} P( Z < z )$.

In addition, we observed earlier that $\phi(\db_0, \z^{-i}) = \phi(\db_1, \z^{-i})$ 
for all $\z^{-i} \in \real^{M-1}$.  Therefore, we can express a lower bound for $P(\alg(\db_1) = H)$ strictly in terms of $\db_0$:
\begin{align*}
&P(\alg(\db_1) = B) \\
&\geq \!\!\int\!\! \indic{ \phi(\db_0, \z^{-i}) } f_{\Z^{-i}}(\z^{-i}) e^{- 2\sens \bcost/\lambda} P(Z_i < \zbnd{\db_0} ) \mathrm{d}\z^{-i} \\
&= e^{- 2\sens \bcost/\lambda} P(\alg(\db_0) = B) = e^{-\epsilon} P(\alg(\db_0) = B)
\end{align*}
since, according the algorithm description, $\lambda = 2\sens \bcost / \epsilon$.
\end{proof}


{\em Remark  \;} In Algorithm~\ref{alg:privpartintervals} we can reduce the noise from $2 \sens \bcost$ to $\sens \bcost$ plus the sensitivity of the particular bucket.  The benefit is a reduction in noise (by at most a factor of 2) for some buckets.  This optimization is used in the experiments.  

\section{\hspace*{-1.2ex}Private Bucket Count Estimation} \label{sec:stats}

This section describes the second stage of the \DW algorithm.  Given the partition $B=\{b_1, \ldots, b_k\}$ determined by the first stage of \DWnospace, it remains to privately estimate counts for each bucket, using budget $\epsilon_2$.  Thus the goal of the second stage is to produce $\s = s_1 \dots s_k$.  Naive solutions like adding Laplace noise to each bucket count result in high error for many workloads.  In this section, we show how to adapt the existing framework of the matrix mechanism \cite{Li:2010Optimizing-Linear} to create a workload-adaptive algorithm for computing the bucket counts.  Within this framework, we describe a novel greedy algorithm for minimizing error of the workload queries.

\subsection{Workload-adaptive bucket estimation}

Our approach relies on the matrix mechanism \cite{Li:2010Optimizing-Linear}, which provides a framework for answering a batch of linear queries (i.e. a workload). Instead of answering the workload directly, the matrix mechanism poses another set of queries, called the query strategy, and uses the Laplace mechanism to obtain noisy answers.  These noisy answers can then be used to derive an estimated data vector using ordinary least squares.  The answers to the workload can then be computed from the estimated data vector.  

Adapting the matrix mechanism to the private estimation of the bucket counts entails two challenges.  First, our original workload $\W$ is expressed in terms of $\x$, but we seek a mechanism that produces estimates of the bucket counts $\s$.  Below, in Sec.~\ref{sec:sub:transform}, we describe a transformation of $\W$ into a new workload in terms of the domain of buckets that allows us to optimize error for the original $\W$.  Second is the familiar challenge of the matrix mechanism: computing a query strategy, suited to the given workload but not dependent on the data, so as to minimize the mean square error of answering the workload.  In general, computing the query strategy that minimizes error under the matrix mechanism requires solving high complexity optimization problems \cite{Li:2010Optimizing-Linear, Yuan12Low-Rank}. Hence we extend ideas from prior work~\cite{chaopvldb12, Yaroslavtsev13Accurate} that efficiently compute approximately optimal query strategies.  We fix a template strategy that is well-suited for anticipated workloads, but then compute approximately optimal weighting factors to emphasize the strategy queries that matter most to the workload.  This effectively adjusts the privacy budget to maximize accuracy on the given workload.  Since our anticipated workload consists of range queries, we use a hierarchical query strategy as a template, similar to prior work~\cite{hay2010boosting, Cormode11Differentially, xiao2010differential}.

Our goal is to minimize the mean squared error of answering the workload by assigning different scaling to queries. Although similar goals are considered in prior works, their methods impose additional constraints that do not apply in our setting: Li and Miklau~\cite{chaopvldb12} require the number of strategy queries to be no more than the domain size; and Yaroslavtsev et al.~\cite{Yaroslavtsev13Accurate} require a fixed ``recovery'' matrix to derive workload answers from strategy answers.

\subsection{Workload transformation} \label{sec:sub:transform}

Recall that, given statistics $\s=s_1 \dots s_k$ for the buckets in $B$, we answer the workload $\W$ by first uniformly expanding $\s$ into an estimate $\estx$ and then computing $\W(\estx)$.  We now show an equivalent formulation for this process by transforming the $m \times n$ workload $\W$ into a new $m \times k$ workload $\tW$ such that the workload query answers can be computed directly as $\tW(\s)$.  The most important consequence of this transformation is that we need not consider the uniform expansion step while adapting the estimation of $\s$ to the workload.  Since $k<n$, an additional convenient consequence is that the domain size is typically reduced so that the complexity of the matrix mechanism operations is lower\footnote{Table \ref{tab:optimalk} in Sec.~\ref{sec:experiments} shows the domain sizes of the derived partitions for each example dataset we consider in the performance analysis.}.

\begin{definition}[Query transformation]
Given a query $q=(q_1, \ldots, q_n)$, defined on data vector $\x$, and given partition $\B$, the transformation of $q$ with respect to $\B$ is defined as $\tq=(\tq_1,\ldots, \tq_k)$ where
$$\tq_j = \frac{1}{|b_j|} \sum_{i\in b_j} q_i.$$
\end{definition}

\begin{example}
 Recall the partition $B=\{b_1, b_2, b_3, b_4\}$ in Fig.~\ref{fig:example}. Consider range query $q=x_2+x_3+x_4+x_5+x_6$. This query can be reformulated in terms of the  statistics of the three buckets $b_1$ to $b_3$ that cover the range spanned by $q$.
 Hence $\tq = \frac{1}{2}s_1 + s_2 + \frac{3}{4}s_3$.
 \end{example}
 
Accordingly, a transformed workload, $\tW$, is formed by transforming each query in $\W$.

\begin{restatable}{proposition}{propqueryconvert}\label{prop:queryconvert}
For any workload $\W$ and buckets $\B$, let $\tW$ be the transformation of $\W$.  Then for any statistics $\s$ for $\B$, 
$$\W(expand(\B,\s))=\tW(\s)$$
\end{restatable}


We now seek a private estimation procedure for the bucket counts in $\s$ that is adapted to the transformed workload $\tW$.  It follows 
from Prop.~\ref{prop:queryconvert} that minimizing the error for workload $\tW$ will also result in minimum error for $\W$ after expansion.

\subsection{Optimal scaling of hierarchical queries} \label{sec:budget}

As mentioned above, our template strategy, denoted as $Y$, is a set of interval queries that forms a tree with branching factor $t$.  
The queries at the leaves are individual entries of $\s$.  For each higher level of the tree, interval queries of $t$ nodes in the previous level are aggregated, and the aggregated query becomes their parent node in the upper level. Since the number of nodes at each level may not be a multiple of $t$, the last nodes at each level are allowed to have fewer than $t$ children.  This aggregating process is repeated until the topmost level only has one node, whose interval is the entire domain $\s$. 


For each query $q\in Y$, let $c_q$ be the scaling $q$, and $Y_c$ be the set of queries in $Y$ after the scaling.  The goal of this stage of our algorithm is to find a scaling of $Y_c$ to minimize the total squared error of answering $\tW$ using $Y_c$. According to \cite{Li:2010Optimizing-Linear}, scaling up all queries in $Y_c$ by a positive constant does not change the mean squared error of answering any query using $Y_c$. Thus, without loss of generality, we bound the sensitivity of $Y_c$ by 1 by requiring that,
\begin{equation}\label{eqn:epsbound}
\sum_{q(i)\neq 0, q\in Y}c_q \leq 1,\quad i=1,\ldots, k.
\end{equation}

When the sensitivity of $Y_c$ is fixed, the scaling controls the accuracy with which the query will be answered: the larger scaling leads to the more accurate answer. 
Let the matrix representation of $Y$ be $\Y$ and $\D_Y$ be the diagonal matrix whose diagonal entries are scales of queries in $Y_c$. Then the matrix form of $Y_c$ can be represented as $\D_Y\Y$. Since the sensitivity of $Y_c$ is bounded by $1$, according to \cite{Li:2010Optimizing-Linear}, the squared error of answering a query $q$ (with matrix form $\q$) using $Y_c$ under $\epsilon_2$ differential privacy is:
$$\frac{2}{\epsilon_2^2}\q^T((\D_Y\Y)^T\D_Y\Y)^{-1}\q=\frac{2}{\epsilon_2^2}\q^T(\Y^T\D_Y^2\Y)^{-1}\q.$$
Let $\tWW$ be the matrix form of $\tW$. The total squared error of answering all queries in $\W$ can then be computed as:
\begin{align}
\sum_{q\in \tW}\frac{2}{\epsilon_2^2}\q^T(\Y^T\D_Y^2\Y)^{-1}\q&=\frac{2}{\epsilon_2^2}tr(\tWW(\Y^T\D_Y^2\Y)^{-1}\tWW^T)\nonumber\\
&=\frac{2}{\epsilon_2^2}tr(\tWW^T\tWW(\Y^T\D_Y^2\Y)^{-1}).\label{eqn:mse}
\end{align}
Above, $tr()$ is the trace of a square matrix: the sum of all diagonal entries.  We now formally state the query scaling problem.

\begin{problem}[Optimal Query Scaling Problem] \label{prob:budget}
The \\optimal query scaling problem is to find a $c_q$ for each query $q\in Y$ that minimizes Eqn~(\ref{eqn:mse}) under the constraint of Eqn~(\ref{eqn:epsbound}).
\end{problem}

\subsection{Efficient greedy scaling} \label{sec:greedy_scaling}
Problem \ref{prob:budget} is a significant simplification of the general strategy selection problem from \cite{Li:2010Optimizing-Linear} because the template strategy is fixed and only scaling factors need to be computed. Nevertheless the optimal solution to Problem~\ref{prob:budget} appears difficult since equation~(\ref{eqn:mse}) is still non-convex.  Instead of pursuing an optimal solution, we solve the problem approximately using the following greedy algorithm. The algorithm works in a bottom-up manner. It initially puts scale $1$ to all leaf queries in $Y$ and $0$ to all other queries in $Y$.  For query $q\in Y$, the algorithm chooses a $\lambda_q\in[0, 1]$. The scaling is then reallocated as follows: the scaling on each of its descendent $q'$ is reduced from $c_{q'}$ to $(1-\lambda_q)c_{q'}$ and the scaling on $q$ is $\lambda_q$. The value of $\lambda_q$ is chosen to minimize Equation~(\ref{eqn:mse}) after the scaling reallocation. Notice that the new scaling still satisfies the constraint in equation~(\ref{eqn:epsbound}).

\begin{figure}[t]
\centering
\includegraphics[width=0.45\textwidth]{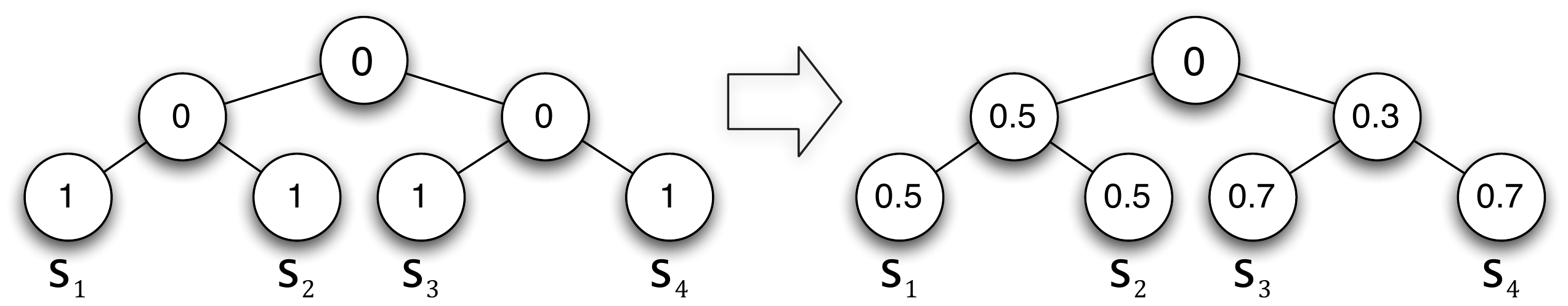}
\caption{Scaling allocation for the second level. $\lambda=0.5$ at the left node, and $\lambda=0.3$ at the right node.\label{fig:greedyexp}}
\end{figure}

\begin{example}
An example of the scaling reallocation is shown in Fig.~\ref{fig:greedyexp}, in which two different $\lambda$ are chosen for two nodes (queries) at the second level.
\end{example}

When the scaling reallocation terminates, the algorithm asks all scaled queries and adds $\Lap(1/\epsilon_2)$ noise to the answer of each query. After that, any inconsistencies  among those noisy query answers are resolved using ordinary least squares inference. 

The major challenge in the algorithm described above is to efficiently choose $\lambda_q$ in each step. To simplify the presentation, we always assume the branching factor $t=2$, though the discussion is valid for any branching factor.

For each interval query $q\in Y$, let $[i, j]$ be the corresponding interval of $q$. Use $\tWW_q$ to denote the matrix consisting of the $i^{th}$ to $j^{th}$ column of $\tWW$, and $\Y_q$ to denote the matrix consisting of the $i^{th}$ to $j^{th}$ column of the matrix of queries in the subtree of rooted at $q$. Let  $\D_{q}$ be the diagonal matrix whose diagonal entries are $c_q'$ for all $q'$ in the subtree rooted at $q$. For each query $q\in Y$ that is not on a leaf of $Y$, let $q_1, q_2$ be queries of its child nodes.

For each query $q\in Y$ that is not on a leaf of $Y$, according to the construction of $Y$, $q=q_1+q_2$. Hence $\tWW_q = [\tWW_{q_1}\: \tWW_{q_2}]$. Futher, since the queries in the subtree of $q$ are the union of queries in the subtree of $q_1$, $q_2$, as well as query $q$ itself, for a given $\lambda_q$,
$$\D_{q} = \left[\begin{smallmatrix} \lambda_q & 0 & 0 \\  0 & (1-\lambda_q)\D_{q_1} & 0 \\ 0 & 0 & (1-\lambda_q)\D_{q_2} \end{smallmatrix}\right].$$

When choosing a $\lambda_q$, due to the fact that the scalings on all ancestors of $q$ in $Y$ are $0$ at this moment, the matrix $\Y^T\D^2_{Y}\Y$ becomes a block diagonal matrix,  and $\Y^T_q\D^2_{q}\Y_q$ is one of its blocks. Therefore, the choice of $\lambda_q$ only depends on $\tWW_q$ and $Y_q$, which means $\lambda_q$ can be determined locally, by minimizing 
\begin{equation}\label{eqn:partialmse}
tr({\tWW_q}^T\tWW_q(\Y^T_q\D^2_{q}\Y_q)^{-1}).
\end{equation}
Since the only unknown variable in Eqn.~(\ref{eqn:partialmse}) is $\lambda_q$, solving its optimal solution is much easier than  solving the optimal scaling for all queries in Eqn.~(\ref{eqn:mse}). However, one of the problems of choosing $\lambda_q$ using equation~(\ref{eqn:partialmse}) is that it is biased towards $q$ and $\lambda_q$ is larger than required. When deciding the scaling on a query $q\in Y$, the scalings on all the ancestors of $q$ are $0$. Hence the scaling distribution is based on the assumption that all queries that contain $q$ are answered by $q$, which is not true after some ancestors of $q$ are assigned non-zero scalings. 

In order to reduce this bias, a heuristic decay factor $\mu$ is introduced to control the impact of $q$ on queries that need to be answered with $q$. The following matrix is used in equation~(\ref{eqn:partialmse}) to take the place of ${\tWW_q}^T\tWW_q$:
\begin{equation}\label{eqn:decaymat}
\mu{\tWW_q}^T\tWW_q + (1-\mu)\left[\begin{smallmatrix}{\tWW_{q_1}}^T\tWW_{q_1} & 0 \\ 0 & {\tWW_{q_2}}^T\tWW_{q_2}\end{smallmatrix}\right].
\end{equation}
As above, the bias of equation~(\ref{eqn:partialmse}) comes from the assumption that the scalings on all the ancestors of $q$ are $0$. Hence there will be less bias when $q$ is more close to the root of $Y$. In our implementation, $\mu$ is set to be $t^{-\frac{l}{2}}$ where $t$ is the branching factor of $Y$ and $l$ is the depth of $q$ in $Y$. Our algorithm then minimizes the following quantity instead of equation~(\ref{eqn:partialmse}).
\begin{equation}\label{eqn:decaypartialmse}
{\scriptstyle tr\left(\left(t^{-\frac{l}{2}}{\tWW_q}^T\tWW_q + (1-t^{-\frac{l}{2}})\left[\begin{smallmatrix}{\tWW_{q_1}}^T\tWW_{q_1} & 0\\ 0 & {\tWW_{q_2}}^T\tWW_{q_2}\end{smallmatrix}\right]\right)(\Y^T_q\D^2_{q}\Y_q)^{-1}\right).}
\end{equation}

At first glance, computing equation~(\ref{eqn:decaypartialmse}) seems complicated since $(\Y^T_q\D^2_{q}\Y_q)^{-1}$ needs to be recomputed for each $\lambda_q$. However, if we record some quantities in the previous step, it only takes $O(m(j-i+1)+(j-i+1)^2)$ time for the preprocessing and Eqn.~(\ref{eqn:decaypartialmse}) can be computed in $O(1)$ time for any $\lambda_q$. 
The detailed proof can be found in Appendix~\ref{sec:app:proof4}.

\begin{algorithm}[t]
\small
\caption{Estimating bucket counts $\s$.}
\label{alg:greedyhier}
    \begin{algorithmic}
    \Procedure{BucketCountEstimator}{$B$, $W$, $\x$, $\epsilon_2$}
    \State Given workload $W$ and buckets $B$, transform workload to $\tW$
    \State Let $Y$ be a tree of queries over buckets
    \State For each query $q\in Y$, let $c_q=1$ if $q$ is a leaf, and $c_q=0$ otherwise.
    \ForAll{$q\in Y$, from bottom to top}
    	\State Numerically find $\lambda_q$ that minimizing Equation~(\ref{eqn:decaypartialmse}).
	\State Let $c_q = \lambda_q$.
	\State For each descendent $q'$ of $q$, let $c_{q'}=(1-\lambda_q)c_{q'}$.
    \EndFor
    \State Let $\y$ be the vector of $c_qq(B(\x)) + \Lap(1/\epsilon_2)$ for all $q\in Y$.
    \State \Return $\s = (\Y^T\D_Y^2\Y)^{-1}(\D_{Y}\Y)^T\y$
    \EndProcedure
    \end{algorithmic}
\end{algorithm}

The entire process of computing bucket statistics is summarized in Algorithm~\ref{alg:greedyhier}. Notice that Algorithm~\ref{alg:greedyhier} just chooses a scaling $Y_c$ with sensitivity at most $1$, and answers $\s$ using $Y_c$ as a strategy. The privacy guarantee of Algorithm~\ref{alg:greedyhier} follows from that of the matrix mechanism.

\begin{proposition}Algorithm~\ref{alg:greedyhier} is $\epsilon_2$-differentially private.
\end{proposition}

\vspace*{-3ex}
\begin{restatable}{theorem}{thmgreedycomplex}\label{thm:greedycomplex}
Algorithm~\ref{alg:greedyhier} takes $O(mk\log k+ k^2)$ time. In the worst case, $k=O(n)$, and Algorithm~\ref{alg:greedyhier} takes $O(mn\log n+ n^2)$ time.
\end{restatable}
\vspace*{-1.5ex}

Hence, Algorithm~\ref{alg:greedyhier} costs much less time than previous general query selection approaches in the matrix mechanism~\cite{Li:2010Optimizing-Linear, Yuan12Low-Rank}.

\section{Experimental Evaluation} \label{sec:experiments}

We now evaluate the performance of \DW on multiple datasets and workloads, comparing it with recently-proposed algorithms (in Sec.~\ref{sec:mainexpt}).  We also examine the effectiveness of each of the two main steps of our algorithm (Sec.~\ref{sec:sub:steps}).  Finally, we consider an extension of our technique to two-dimensional spatial data and compare it with state-of-the-art algorithms (Sec.~\ref{sec:exptspatial}).

\subsection{Experimental setup} \label{expt:setup}

In the experiments that follow, the primary metric for evaluation is the average $L_1$ error  per query for answering the given workload queries. Most workloads we use are generated randomly (as described below). 
Each experimental configuration is repeated on 5 random workloads with 3 trials for each workload.  The results reported are the average across workloads and trials.  The random workloads are generated once and used for all experiments.  

The privacy budget in \DW is set as $\epsilon_1 = 0.25\epsilon$ and $\epsilon_2 = 0.75\epsilon$.  Unless otherwise specified, the first step of \DW constructs a partition using intervals whose lengths must be a power of 2, an approximation that is described in Sec.~\ref{sec:partition}.  For the second step of the algorithm, the branching factor of the query tree is set to 2.
\vspace*{-1.5ex}
\paragraph*{Datasets} There are seven different 1-dimensional datasets considered in our experiments. Although these datasets are publicly available, many of them describe a type of data that could be potentially sensitive, including financial, medical, social, and search data. \textsf{Adult} is derived from U.S. Census data~\cite{Bache+Lichman:2013}: the histogram is built on the ``capital loss'' attribute, which is the same attribute used in \cite{hardt2012a-simple}. \textsf{Income} is based on the IPUMS American community survey data from 2001-2011; the histogram attribute is personal income~\cite{ipums::usa_v5.0}. \textsf{Medical Cost} is a histogram of personal medical expenses based on a national home and hospice care survey from 2007~\cite{icpsr}. \textsf{Nettrace} is derived from an IP-level network trace collected at the gateway router of a major university.  The histogram attribute is the IP address of internal hosts and so the histogram reports the reports the number of external connections made by each internal host~\cite{hay2010boosting}. \textsf{Search Logs} is a dataset extracted from search query logs that reports the frequency of the search term ``Obama'' over time (from 2004 to 2010)~\cite{hay2010boosting}.  Furthermore, we consider two temporal datasets derived from two different kinds of network data.   \textsf{HepPh} is a citation network among high energy physics pre-prints on arXiv and \textsf{Patent} is a citation network among a subset of US patents~\cite{snap}.  These last datasets describe public data but serve as a proxy for social network data, which can be highly sensitive.   
For both datasets, the histogram reports the number of new incoming links at each time stamp. To eliminate the impact of domain size in comparing the ``hardness'' of different datasets, all datasets above are aggregated so that the domain size $n$ is 4096.  

\vspace*{-1.5ex}
\paragraph*{Query workloads} We run experiments on four different kinds of workloads. The \textit{identity} workload consists of all unit-length intervals $[1,1], [2,2], \dots, [n,n]$. The \textit{uniform interval} workload samples 2000 interval queries uniformly at random. In addition, workloads that are not uniformly distributed over the domain are also included. 
%
The \textit{clustered interval} workload first samples five numbers uniformly from $[1,n]$ to represent five cluster centers and then samples 400 interval queries for each cluster.  
 Given cluster center $c$, an interval query is sampled as $[c - |X_{\ell}|, c + |X_{r}|]$ where $X_{\ell}$ and $X_{r}$ are independent random variables from a normal distribution with a standard deviation of $256$.  The \textit{large clustered interval} workload is generated in the same way but the standard deviation is $1024$.

\vspace*{-1.5ex}
\paragraph*{Competing algorithms} We compare \DW with six algorithms. For data-independent algorithms, we include a simple approach (Identity) that adds Laplace noise to each entry of $\x$ and the Privelet algorithm~\cite{xiao2010differential}, which is designed to answer range queries on large domains. For data-dependent algorithms, we compare with EFPA~\cite{Acs2012compression}, P-HP~\cite{Acs2012compression}, StructureFirst~\cite{xu2013differential},\footnote{The other algorithms from Xu et al.~\cite{xu2013differential} take more than 20 hours to complete a single trial.  Therefore, they are not included.} and MWEM~\cite{hardt2012a-simple}, all of which are described in Sec.~\ref{sec:related}. For MWEM, we set the number of iterations, $T$, to the value in $\{10, 20, \dots, 190$, $ 200\}$ that achieves the lowest error on each dataset for the \emph{uniform intervals} workload and $\epsilon=0.1$. We use that $T$ for all experiments on that dataset.
%

With the exception of MWEM, all algorithms are quite efficient, usually finishing in seconds. \!MWEM slows down for harder datasets which require a high $T$, taking up to ten seconds on these datasets.

\subsection{Accuracy for interval workloads} \label{sec:mainexpt}
\begin{figure}[t]
\vspace*{-.5ex}
\includegraphics[width=.45\textwidth]{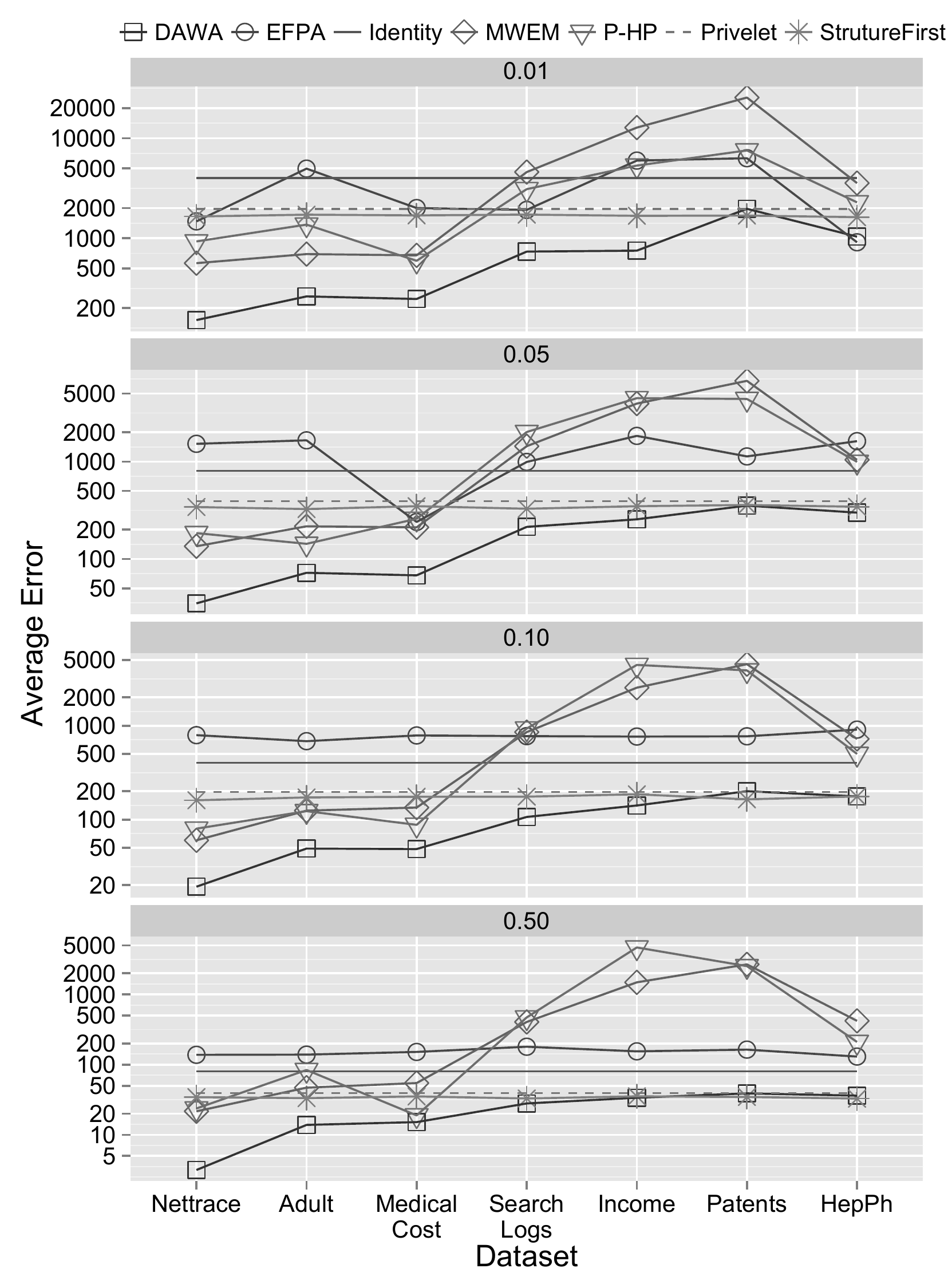}
\caption{\label{fig:allerr} Average error on the {\em uniform intervals} workload across multiple datasets.  The privacy budget ranges from $\epsilon=0.01$ (top) to $\epsilon=0.5$ (bottom). }
\end{figure}

\begin{table}[t]
\small
\centering
\begin{tabular}{|c|c|c|c|c|c|c|}
\hline \textsf{\scriptsize Nettrace} & \textsf{\scriptsize Adult} & \textsf{\scriptsize Med. Cost} & \textsf{\scriptsize S. Logs} & \textsf{\scriptsize Income} & \textsf{\scriptsize Patents} & \textsf{\scriptsize HepPh}\\
\hline
22 & 29 & 20 & 500 & 1537 & 1870 & 2168\\
\hline
\end{tabular}
\caption{The number of buckets, $k$, in the optimal partition when $\epsilon=0.1$.  The original domain size is $n=4096$ for each dataset. \label{tab:optimalk}}
\end{table}

\begin{table}[t]
\centering
\small
\subfigure[Smallest ratio across datasets]{ 
\begin{tabular}{|c||c|c|c|c|c|c|}
\hline $\epsilon$ & {\scriptsize Identity} & {\scriptsize Privelet} & {\scriptsize MWEM} & {\scriptsize EFPA} & {\scriptsize P-HP} & {\scriptsize S. First} \\
\hline 0.01 & 2.04 & 1.00 & 2.65 & 0.88 & 2.20 & 0.86 \\
\hline 0.05 & 2.27 & 1.11 & 3.00 & 3.20 & 1.98 & 1.01 \\
\hline 0.1 & 2.00 & 0.98 & 2.54 & 3.84 & 1.81 & 0.82 \\
\hline 0.5 & 2.06 & 1.01 & 3.39 & 3.60 & 1.25 & 0.89 \\
\hline
\end{tabular} \label{tab:smallest}}
\subfigure[Largest ratio across datasets]{
\begin{tabular}{|c||c|c|c|c|c|c|}
\hline $\epsilon$ & {\scriptsize Identity} & {\scriptsize Privelet} & {\scriptsize MWEM} & {\scriptsize EFPA} & {\scriptsize P-HP} & {\scriptsize S.First} \\
\hline 0.01 & 26.42 & 12.93 & 17.00 & 18.94 & 7.09 & 10.85 \\
\hline 0.05 & 22.97 & 11.24 & 19.14 & 43.58 & 17.57 & 9.77 \\
\hline 0.1 & 20.85 & 10.20 & 22.54 & 41.09 & 31.41 & 8.32 \\
\hline 0.5 & 25.47 & 12.46 & 68.75 & 43.69 & 138.14 & 10.89 \\
\hline
\end{tabular} \label{tab:largest}}
\caption{\label{tab:allerr} Ratio of algorithm error to \DW error, for each competing algorithm and $\epsilon$ setting on \emph{uniform intervals}: \subref{tab:smallest} smallest ratio observed across datasets; \subref{tab:largest} largest ratio across datasets.}
\end{table}

Fig.~\ref{fig:allerr} presents the main error comparison of all algorithms on workloads of \textit{uniform intervals} across a range of datasets and settings of $\epsilon$.  
While data-independent algorithms like Privelet and Identity offer constant error across datasets, the error of data-\\dependent algorithms can vary significantly.\footnote{StructureFirst is an exception to this trend: its observed performance is almost totally independent of the dataset.  Its partition selection algorithm uses a high sensitivity scoring function (which is based on $L_2$ rather than $L_1$).  Thus, partition selection is very noisy and close to random for all datasets.}  For some datasets, data-dependent algorithms can be much more accurate.  For example, on \textsf{Nettrace} with $\epsilon=0.01$, \emph{all} of the data-dependent algorithms have lower error than the best data-independent algorithm (Privelet).  For this dataset, the error of \DW is at least an order of magnitude lower than Privelet.  These results suggest the potential power of data-dependence.

There are other datasets, however, where the competing data-dependent algorithms appear to break down.    In the figure, the datasets are ordered by the cost of an optimal partition (i.e., an optimal solution to Step 1 of our algorithm) when $\epsilon_2=0.1$.  This order appears to correlate with ``hardness.''  Datasets on the left have low partition cost and appear to be relatively ``easy,'' presumably because data-dependent algorithms are able to exploit uniformities in the data.  However, as one moves to the right, the optimal partition cost increases and the datasets appear to get more difficult.  It is on many of the ``harder'' datasets where competing data-dependent algorithms suffer: their error is higher than even a simple baseline approach like Identity.  

In contrast, \DW does not break down when the dataset is no longer ``easy.''  On the moderately difficult dataset \textsf{Search Logs}, \DW is the only data-dependent algorithm that outperforms data-independent algorithms. On the ``hardest'' datasets, its performance is comparable to data independent techniques like Privelet.  \DW comes close to achieving the best of both worlds: it offers very significant improvement on easier datasets, but on hard datasets roughly matches the performance of data-independent techniques.

For the same workload, datasets, and algorithms, Table~\ref{tab:allerr} reports the performance of \DW relative to other algorithms.  Each cell in the table reports the ratio of algorithm error to \DW error.  Table~\ref{tab:allerr}\subref{tab:smallest} reports the smallest ratio achieved over all datasets---i.e., how close the competing algorithm comes to matching, or in some cases beating, \DW\hspace{-2pt}.  Table~\ref{tab:allerr}\subref{tab:largest} reports the largest ratio achieved---i.e., how much worse the competing algorithm can be on some dataset.  Table~\ref{tab:allerr}\subref{tab:largest} reveals that every competing algorithm has at least 7.09 times higher error than \DW on some dataset.  

Table~\ref{tab:allerr}\subref{tab:smallest} reveals that \DW is sometimes less accurate than another algorithm, but only moderately so.  This occurs on the ``hardest'' datasets, {\sf Patents} and {\sf HepPh}, where \DW has error that is at most $\frac{1}{0.82} \approx 22\%$ higher than other approaches.  On these hard datasets, the optimal partition has thousands of buckets (see Table~\ref{tab:optimalk}), indicating that it is highly non-uniform.  On non-uniform data, the first stage of the \DW algorithm spends $\epsilon_1$ of the privacy budget just to select a partition that is similar to the base buckets.  Despite the fact that the first stage of the algorithm does not help much on ``hard'' datasets, \DW is still able to perform comparably to the best data-independent technique, in contrast to the other data dependent strategies which perform poorly on such ``hard'' datasets.

In addition to {\em uniform interval} workload, we also ran experiments on the other three types of workloads.  The performance of \DW relative to its competitors is qualitatively similar to the performance on {\em uniform interval} workload shown above. Due to limited space, the figures are omitted.

\subsection{Analyzing the performance of \DW} \label{sec:sub:steps}

To further understand the strong performance shown above, we study the two steps of \DWnospace in detail, first by isolating the impact of each step, and then by assessing the effectiveness of the approximations made in each step. 

\subsubsection{Isolating the two steps} \label{sec:dawabreak}

\begin{figure}[t]
\centering
\vspace*{-1ex}
\includegraphics[width=.44\textwidth]{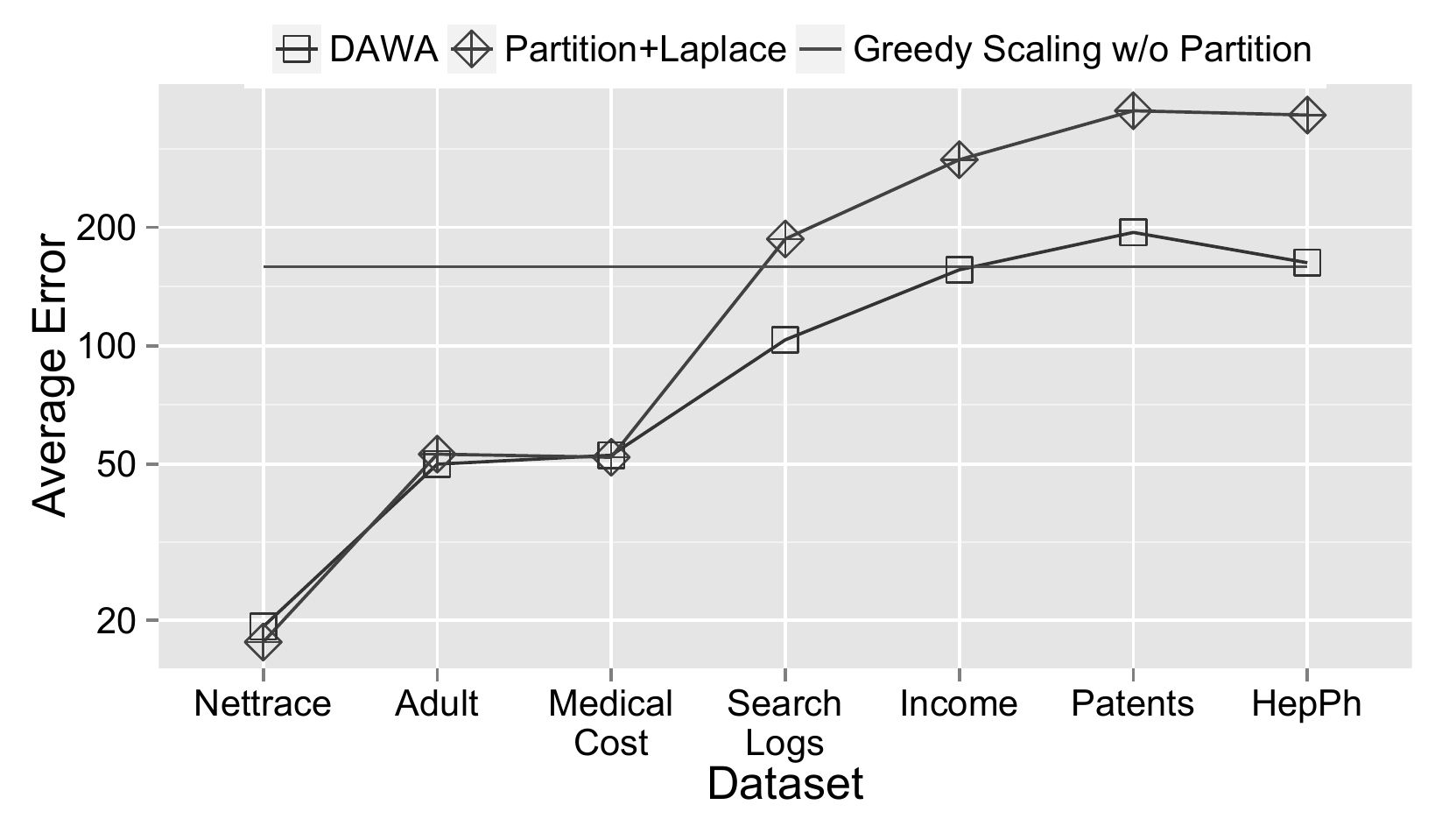}
\caption{Average error of isolated parts of \DW, $\epsilon=0.1$\label{fig:compdawa}.}
\end{figure}

To isolate the performance of each of the two stages of the \DW algorithm, we consider two \DW variants.  The first variant combines the first stage of \DW with the Laplace mechanism (Partition+Laplace).  This algorithm is data dependent but not workload aware.  This variant has two stages like \DWnospace, and the budget allocation is the same: a quarter of the budget is spent on partitioning and the rest on estimating bucket counts.  The second variant omits the first stage of \DW and runs the second stage of the algorithm on the original domain (Greedy Scaling w/o Partition).  For this variant, the entire budget is allocated to estimating counts. This algorithm is workload-aware but data-independent, thus its performance is the same across all datasets.

Fig.~\ref{fig:compdawa} shows the results.  On the ``easier'' datasets, \DW has much lower error than Greedy Scaling w/o Partition.  For these datasets, which have large uniform regions, allocating a portion of the privacy budget to selecting a data-dependent partition can lead to significant reductions in error.  In these cases, the benefit outweighs the cost.
On ``hard'' datasets, where most data-dependent algorithms fail, the partitioning does not appear to help much.  One reason may be that on these datasets even the optimal partition has many buckets (Table~\ref{tab:optimalk}), so the partitioned dataset is not radically different from the original domain.  However, even on these hard datasets, \DW is almost as accurate as Greedy Scaling w/o Partition, suggesting that there is still enough improvement from partitioning to justify its cost.  

Finally, we can examine the effect of the second stage of \DW by comparing \DW against Partition+Laplace.  
On ``easy'' datasets, they perform about the same.  On these datasets, the partition selected in the first stage has a small number of buckets, which means that the input to the second step is a small domain.  Since the Laplace mechanism works well on small domains, the lack of a performance difference is not surprising.  However, on ``harder'' datasets, the partitions produced by the first stage have a large number of buckets and Partition+Laplace performs poorly.  In such cases, using the second step of \DW proves highly beneficial and \DW has much lower error than Partition+Laplace.

\subsubsection{Effectiveness of partition selection} \label{sec:exptpartitions}

\begin{figure}[t]
\centering
\vspace*{-1ex}
\includegraphics[width=.44\textwidth]{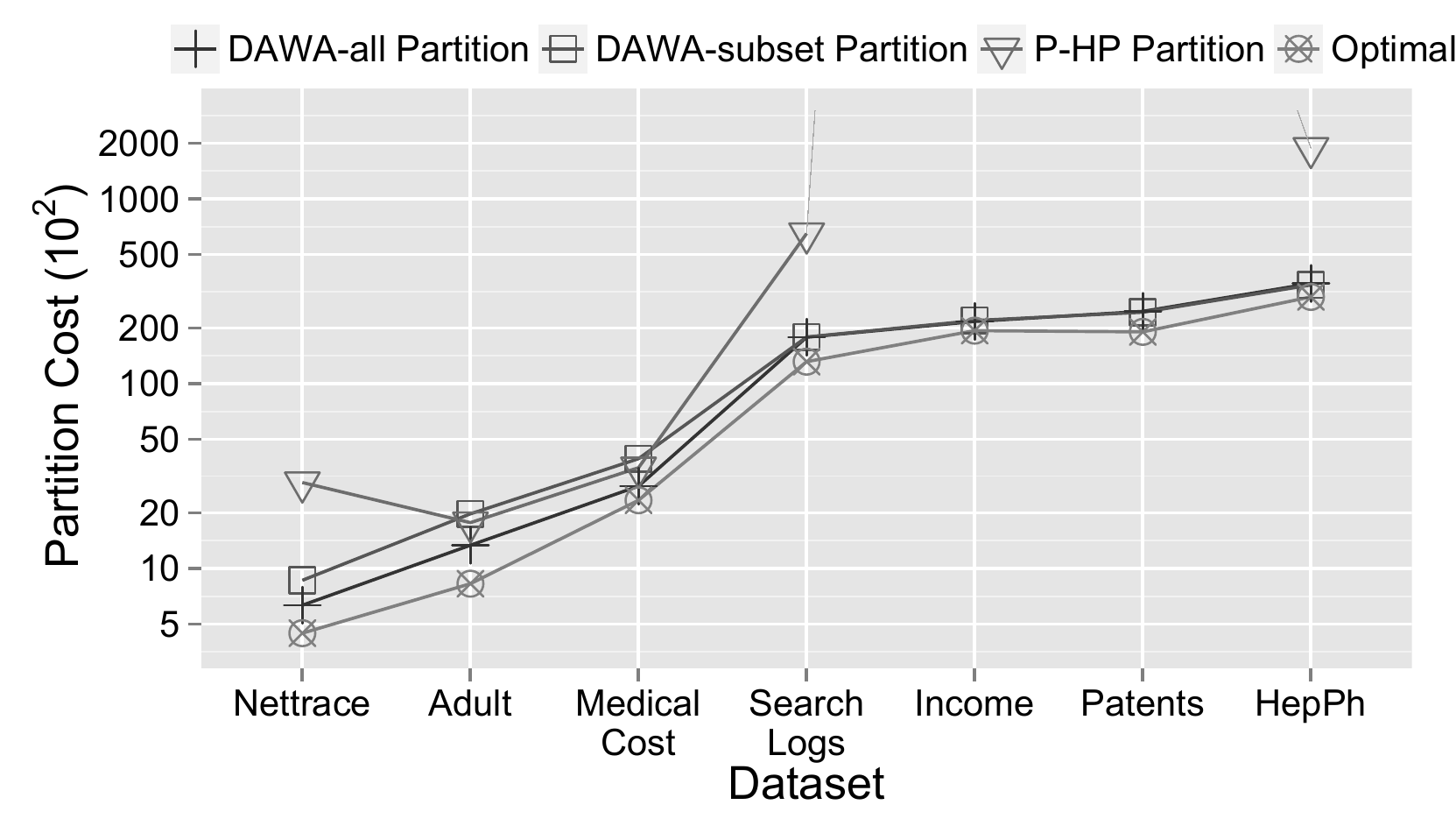}
\caption{A comparison of alternative algorithms for the first step of \DW with $\epsilon=0.1$.\label{fig:vsphp}}
\vspace*{-.5ex}
\end{figure}


Here we evaluate the effectiveness of the first step of the \DW algorithm, partition selection (Algorithm~\ref{alg:privpartintervals}).  Recall from Sec.~\ref{sec:privatepartition} that it is possible to restrict the set of intervals considered in selecting the partition.  We compare two versions of the algorithm: \DWapprox only considers intervals whose lengths are a power of two, \DWall considers all possible intervals.  We compare these variants with the optimal solution, which is computed by solving Problem~\ref{problem:optpartition} using the bucket cost without noise, ignoring privacy considerations.  Finally, we compare with P-HP~\cite{Acs2012compression}, which is also designed to solve Problem~\ref{problem:optpartition}. To facilitate a fair comparison, for this experiment each algorithm spends the same amount of privacy budget on selecting the partition. 

The results are shown in Fig.~\ref{fig:vsphp} for $\epsilon=0.1$ where the y-axis measures the partition cost (Def.~\ref{def:l1cost}). We further assume $\epsilon_2=0.1$ when computing the cost of each bucket in \DW. The partition cost of \DWall is close to optimal.  The cost of the partition of \DWapprox is sometimes higher than that of \DWall especially on ``easier'' datasets.  Generally, however, \DWapprox and \DWall perform similarly.  This suggests that the efficiency benefit of \DWapprox does not come at the expense of utility.  The cost of the partition selected by P-HP is almost as low as the cost of the \DWapprox partition on the \textsf{Adult} and \textsf{Medical Cost} datasets, but it is orders of magnitude larger on other datasets (on \textsf{Income} and \textsf{Patents} it is at least $1.6\times 10^6$).  This provides empirical evidence that Algorithm~\ref{alg:privpartintervals} is much more accurate than the recursive bisection approach of P-HP.  The results with $\epsilon \in \set{0.01, 0.05, 0.5}$ are similar and omitted.

\subsubsection{Effectiveness of adapting to workload}  \label{sec:exptgreedy}
\begin{figure}[t]
\includegraphics[width=.45\textwidth]{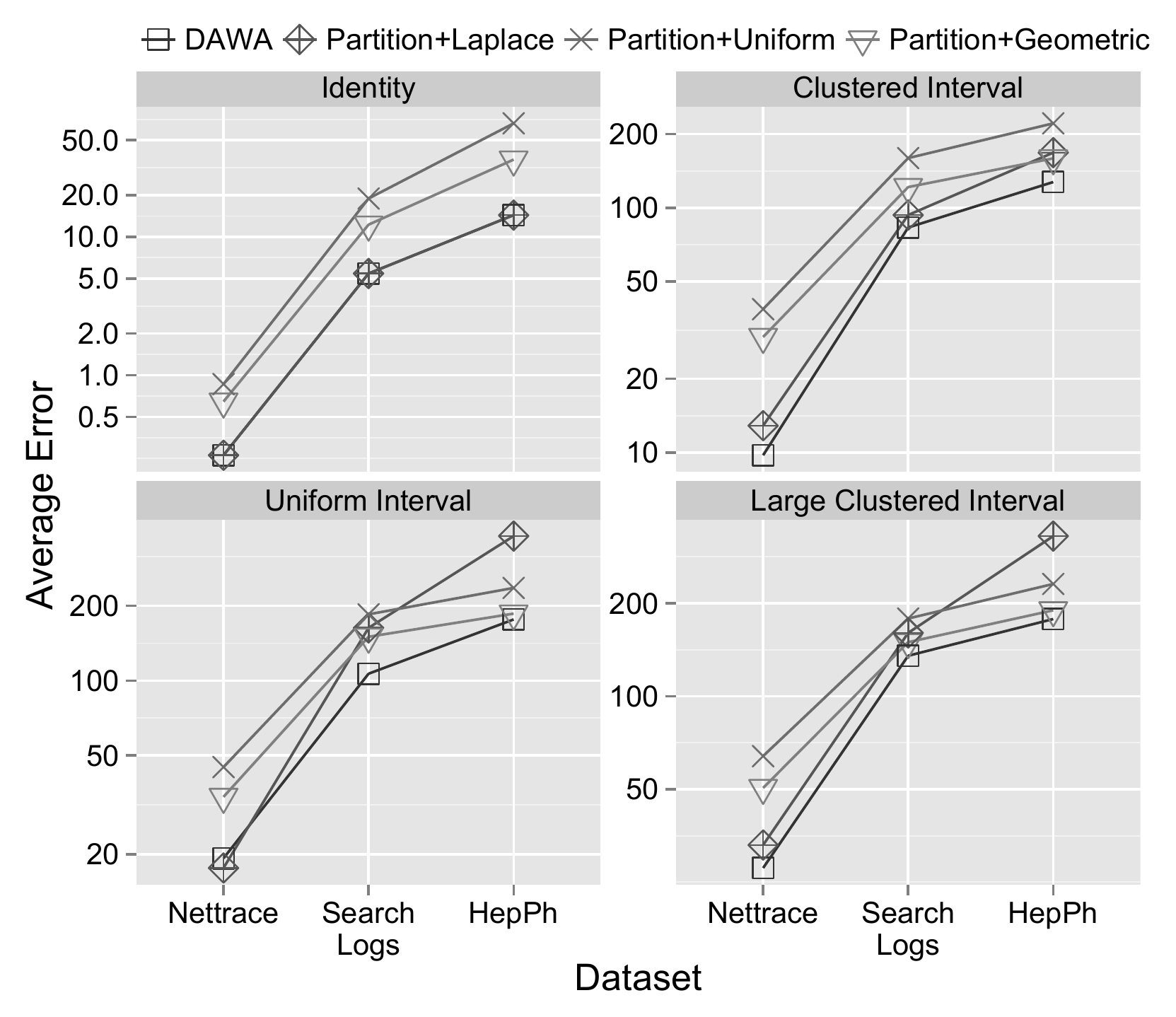}
\caption{A comparison of alternative algorithms for the second step of \DW across different workloads, with $\epsilon=0.1$\label{fig:compsec}. }
\end{figure}

The second stage of \DW designs a query strategy that is tuned to the workload, as described by Algorithm~\ref{alg:greedyhier}.  Here we combine our partitioning algorithm in the first step with some alternative strategies and compare them with \DW to evaluate the effectiveness of our greedy algorithm. Two alternative ways to scale queries in $Y$ are considered: all queries are given the same scaling (Partition+Uniform) based on Hay et al.~\cite{hay2010boosting}, and the scaling decreases geometrically from leaves to root (Partition+Geometric) based on Cormode et al.~\cite{Cormode11Differentially}. The Laplace mechanism (Partition+Laplace) is also included. Among the alternative algorithms, Partition+Geomet\-ric is designed to answer \textit{uniform interval} workloads, and the La-place mechanism is known to be the optimal data-independent mechanism for the \textit{identity} workload.  We do not consider any data-dependent techniques as alternatives for the second step.  After partitioning, uniform regions in the data have been largerly removed and our results show that the data-dependent algorithms perform poorly if used in this step.


Fig.~\ref{fig:compsec} shows results for four different workloads and three different datasets at $\epsilon = 0.1$.  The datasets span the range of difficulty from the ``easier'' \textsf{Nettrace} to the ``harder'' \textsf{HepPh}.  (The algorithms being compared here are affected by the dataset because they operate on the data-dependent partition selected in the first stage.)
The original \DW performs very well on all cases. In particular, it always outperforms Partition+Geometric on \textit{uniform interval} workload and performs exactly same as the Partition+Laplace mechanism on \textit{identity} workload. In the latter case, we find that the greedy algorithm in the second step outputs the initial budget allocation, which is exactly same as the Laplace mechanism. 
\vspace*{-.15ex}

\subsection{Case study: spatial data workloads} \label{sec:exptspatial}
\begin{figure}[t]
\centering
\includegraphics[width=.44\textwidth]{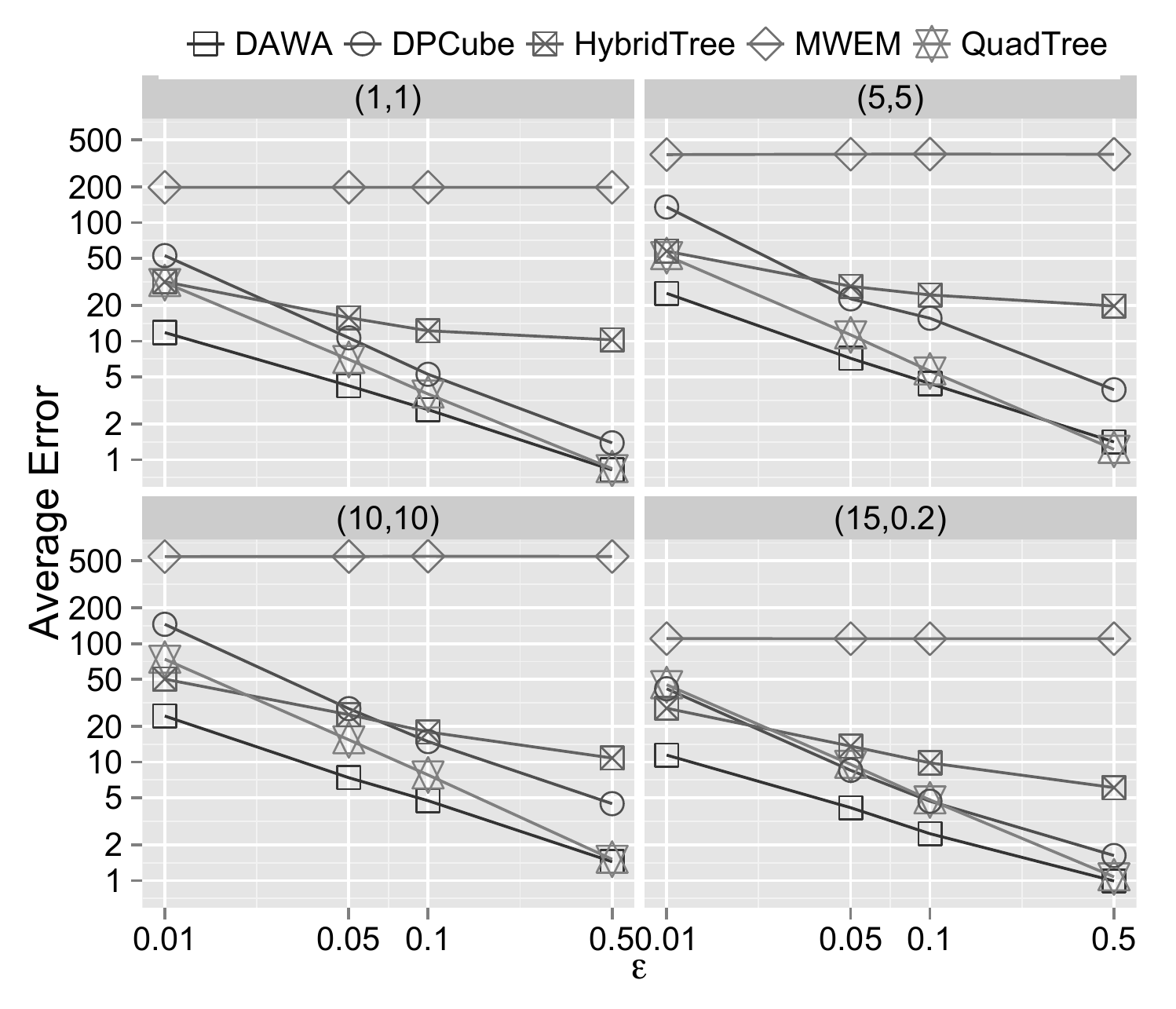}
\caption{\label{fig:2d}Average error answering query workloads on spatial data.  Each workload is a batch of random rectangle queries of a given $(x, y)$ shape.}
\label{fig:tiger2d}
\end{figure}

Lastly, we evaluate an extension to our main algorithm to compute histograms over two dimensional spatial data.  We use an experimental setup that is almost identical to previous work~\cite{Cormode11Differentially}; differences are highlighted below.  The dataset describes the geographic coordinates (latitude and longitude) of road intersections, which serve as a proxy for human population, across a wide region in the western U.S.~\cite{tiger}.  
Over this region, we generate a workload of random rectangle queries of four different shapes: $(1, 1)$, $(5, 5)$, $(10, 10)$, and $(15, 0.2)$ where shape $(x, y)$ is a rectangle that covers $x$ degrees of longitude and $y$ degrees of latitude. 

We compare with a data-independent algorithm, QuadTree~\cite{Cormode11Differentially}, and data-dependent algorithms MWEM, HybridTree~\cite{Cormode11Differentially}, and \\DPCube \cite{Xiao:2012fk}.
Among these algorithms, only MWEM is workload-aware.  Since some algorithms expect discrete domains as input, we discretize the domain by partitioning the space into the finest granularity used by the QuadTree, whose height is $10$~\cite{Cormode11Differentially}.  Thus, both longitude and latitude are split evenly into $2^{10}$ bins. 

To extend the \DW algorithm to two dimensional data, we use a Hilbert curve of order 20 to convert the $2^{10}\times 2^{10}$ grid into a 1-dimensional domain with size $2^{20}$.  In case the query region only partially covers some bins in the discretized domain, the query answer is estimated by assuming uniformity within each bin. 

Fig.~\ref{fig:2d} shows the results.  Although \DW is designed for interval workloads on 1-dimensional data, it performs as well or better than algorithms specifically designed to support rectangular range queries on 2-dimensional data.  The performance gap between\\ \DW and its competitors increases as $\epsilon$ decreases.  


\section{Related work}
\label{sec:related}

A number of data-dependent algorithms have been developed recently~\cite{Acs2012compression,Cormode11Differentially,hardt2012a-simple,Xiao:2012fk,xu2013differential}. We empirically compare \DW against these approaches in Sec.~\ref{sec:experiments}.  
Our algorithm is most similar to P-HP~\cite{Acs2012compression} and Structure First~\cite{xu2013differential}, both of which find a partition and then compute statistics for each bucket in the partition.  P-HP's approach to partitioning is based on the same optimization as presented here (Problem~\ref{problem:optpartition}).  It uses the exponential mechanism to recursively bisect each interval into subintervals.  A key distinction is that P-HP is an approximation algorithm: even if $\epsilon \rightarrow \infty$, it may not return the least cost partition.  In contrast, we show that the optimization problem can be solved directly by simply using noisy scores in place of actual scores and we prove that a constant amount of noise is sufficient to ensure privacy.  The experiments in Sec.~\ref{sec:exptpartitions} show that P-HP consistently returns higher cost partitions than our approach.  Structure First~\cite{xu2013differential} aims to solves a different optimization problem (with a cost function based on $L_2$ rather $L_1$).  In addition, it requires that the user specify $k$, the number of buckets, whereas \DW automatically selects the best $k$ for the given dataset.  Neither P-HP nor StructureFirst is workload-aware.

The other data-dependent mechanisms use a variety of different strategies. The EFPA~\cite{Acs2012compression} algorithm transforms the dataset to the Fourier domain, samples noisy coefficients, and then transforms back.  DPCube~\cite{Xiao:2012fk} and Hybrid Tree~\cite{Cormode11Differentially}, both of which are designed for multi-dimensional data, construct estimates of the dataset by building differentially private KD-trees.  MWEM~\cite{hardt2012a-simple} derives estimate of the dataset iteratively: each iteration selects a workload query, using the exponential mechanism, and then updates its estimate by applying multiplicative weights given a noisy answer to the query.  MWEM supports the more general class of linear queries, whereas \DW is designed to support range queries.  MWEM also offers strong asymptotic performance guarantees.  \\ However, on workloads of range queries, we find in practice that MWEM performs poorly except when the dataset is highly uniform. It is also limited by the fact that it can only asks workload queries, which may not be the best observations to take.

General data-dependent mechanisms are proposed in the theory community~\cite{DBLP:conf/focs/DworkRV10,Gupta:2012uq}.  They are not directly comparable because they work on a slightly weaker variant of differential privacy and are not computationally efficient enough to be practical.

Data-independent mechanisms attempt to find a better set of measurements in support of a given workload, then apply the Laplace mechanism and inference to derive consistent estimates of the workload queries.  Our \DW algorithm would be similar to these methods if the partitioning step always returned the trivial partition, which is $\x$ itself.  Many of these techniques fall within the matrix mechanism framework~\cite{Li:2010Optimizing-Linear}, which formalizes the measurement selection problem as a rank-constrained SDP.  While the general problem has high computational complexity, effective solutions for special cases have been developed.  To support range queries, several mechanisms employ a hierarchical strategy~\cite{Cormode11Differentially,hay2010boosting,xiao2010differential}.  Our approach builds on this prior work. A key difference is that our algorithm adapts the strategy to fit the specific set of range queries given as a workload, resulting in lower workload error.  Other strategies have been developed for marginal queries~\cite{barak2007privacy, Ding11Differentially}.  Yuan et al.~\cite{Yuan12Low-Rank} revisit the general problem for the case when workloads are small relative to the domain size; however the algorithm is too inefficient for the domain sizes we consider here.  Other algorithms have been developed that adapt to the workload.  However, they are not directly applicable  because they are designed for a weaker variant of differential privacy~\cite{chaopvldb12}, or employ a user-specified ``recovery'' matrix, instead of ordinary least squares~\cite{Yaroslavtsev13Accurate}.

\section{Conclusion \& Future Work}
$\DW$ is a two-stage, data- and workload-aware mechanism for answering sets of range queries under differential privacy. \DW first partitions the domain into approximately uniform regions and then derives a count for each region using measurements of varying accuracy that are tuned to the workload queries.  Experimentally, \DW achieves much lower error than existing data-dependent mechanisms on datasets where data-dependence really helps.  On complex datasets, where competing data-dependent techniques suffer, \DW does about the same or better than data-independent algorithms. In this sense, \DW achieves the best of both worlds.



Our results have shown that, for some datasets, data-aware algorithms can reduce error by a factor of 10 or more over competitive data-independent techniques.  But it remains difficult to characterize exactly the properties of a dataset that permit lower error under differential privacy.  Optimal partition cost of a dataset provides some insight into dataset ``hardness'' for the \DW algorithm, but we are not aware of a general and fully-satisfying measure of dataset complexity.  We view this as an important direction for future work.   We would also like to consider extensions to private partitioning that would directly incorporate knowledge of the workload and to extend  our method to a larger class of worloads beyond one- and two-dimensional range queries. 

\paragraph*{Acknowledgements} We are grateful for the comments of the anonymous reviewers. This work is supported by the NSF through grants CNS-1012748, IIS-0964094, CNS-1129454, and by CRA/\-CCC through CI Fellows grant 1019343.

\label{sec:conclusion}

\bibliographystyle{abbrv}
\bibliography{refs}
\newpage
\onecolumn
\appendix
\section{Detailed algorithm description}
This section presents the detailed implementation of subroutines in Algorithm~\ref{alg:privpartintervals}, including {\sc AllCosts} and {\sc LeastCostPartition}.

\subsection{Computing bucket costs efficiently} \label{sec:computecosts}

This section describes the {\sc AllCosts} algorithm which computes the cost for all buckets.  The algorithm takes only $O(\log n)$ time per bucket.  A pseudocode description is shown in Algorithm~\ref{alg:lenscore}.

\begin{algorithm}[t]
\caption{Compute costs for all interval buckets: returns cost array where $cost[b] = \bcost(\x, b)$ for $b=[j_1,j_2]$ for all $1 \leq j_1 \leq j_2 \leq n$.}
\label{alg:lenscore}
    \begin{algorithmic}
    \Procedure{AllCosts}{$\db$, $\epsilon_2$} 
    \State $cost \gets []$
    \For {$k=1, \ldots, n-1$}
        \State Update $cost$ with costs from \Call{CostsSizeK}{$\db$, $\epsilon_2$, $k$}
    \EndFor
    \State \Return $cost$
    \EndProcedure  \\

    \Procedure{CostsSizeK}{$\db$, $\epsilon_1$, $\epsilon_2$, $k$}
    \State $cost \gets []$
    \State Let $T$ be an empty binary search tree with root $t_{root}$.
    \State Insert $x_1, \ldots, x_k$ to $T$.
    \State $sum\leftarrow x_1+\ldots+x_k$.
    \For{$j=k,\ldots,n$}
        \State $i \gets j - k + 1$
        \hspace{4em}\text{// interval $[i,j]$ has length $k$}
    	\State $cost[[i, j]] \gets 1/\epsilon_2 + 2 \;\cdot$ \Call{DevTree}{$t_{root}$, $\frac{s}{k}$}
	\State Remove $x_{i}$ from $T$
	\State Insert $x_{j+1}$ into $T$
	\State $sum\leftarrow sum - x_{i} + x_{j+1}$
    \EndFor
    \State \Return $cost$
    \EndProcedure \\
    
    \Procedure{DevTree}{$t_{root}, a$}
    \State // Compute $\sum_{t\in T, x_t \geq a} (x_t-a)$ on a binary search tree T
    \If{$t_{root}$ is NULL} \Return 0\EndIf
    \If{$x_{root} < a$}
    	\State $r\leftarrow$ \Call{DevTree}{$t_{r}$, $a$} 
    \Else
    	\State $r\leftarrow$ \Call{DevTree}{$t_{\ell}$, $a$} - $(c_{r}+1)a$ +  $\Sigma_{r}$  + $x_{root}$
    \EndIf
    \State \Return $r$
    \EndProcedure
    \end{algorithmic}
\end{algorithm}

When compared to computing costs for a v-optimal histogram (which is based on an $L_2$ metric), computing the costs for the $L_1$ metric used in this paper is more complicated because the cost can not be easily decomposed into sum and sum of squares terms.  A naive computation of the cost would require $\Theta(n)$ per bucket.  However, we show that the cost can be computed through simple operations on a balanced binary search tree and thus it is possible to compute the cost in $O(\log n)$ time per bucket.

The computationally challenging part of computing bucket costs (as defined in Def.~\ref{def:l1cost}) is computing the absolute deviation $dev(\x, b_i)$ for all intervals $b_i$.  Recall that
\begin{align*}
dev(\x, b_i)&=2 \sum_{j\in I^{+}}x_j -|I^{+}|\cdot \bmean
\end{align*}

For interval $b_i=[i_1, i_2]$, the total deviation can be computed knowing only $|I^{+}|$, the number of $x_j$ who are larger than $\bmean$, and the sum of $x_j$ for $j \in I^{+}$.

Those quantities can be efficiently computed with a binary search tree of $x_{j_1}, \ldots, x_{j_2}$. For each node $t$ in the binary search tree, record its value ($x_t$).  In addition, each node $t$ stores the sum of all values in its subtree ($\Sigma_t$), and the number of nodes in its subtree ($c_t$).  For any constant $a$ and any binary search tree $T$, we can then compute $\sum_{t\in T, x_t \geq a} (x_t - a)$.  This what the {\sc DevTree} procedure in Algorithm~\ref{alg:lenscore} computes.

The correctness of {\sc DevTree} can be proved inductively on the height of the tree. It is clear that the answer is correct when $T$ is an empty tree. When $T$ is not empty, let $\ell$ and $r$ denote the left and right children, respectively, and let $T_l$ and $T_r$ denote the left and right subtree, respectively. If $x_{root} < a$, all entries in the left subtree are less than $a$ as well. Thus,
\begin{align*}
\sum_{t\in T, x_t\geq a} ( x_t - a ) &= \sum_{t\in T_r, x_t\geq a} ( x_t -a ).
\end{align*}
When $x_{root}\geq a$, all entries in the right subtree are larger than or equal to $a$ as well. Hence,
\begin{align*}
\sum_{t\in T, x_t\geq a}  ( x_t -a ) &= \sum_{t\in T_l, x_t\geq a} ( x_t -a ) +\sum_{t\in T_r} ( x_t -a ) + ( x_{root} -a )\\
&=\sum_{t\in T_l, x_t\geq a} ( x_t -a ) - (|T_r|+1)a +\sum_{t\in T_r} x_t + x_{root} \\
&=\sum_{t\in T_l, x_t\geq a} ( x_t -a ) - (c_r+1)a + \Sigma_{r} + x_{root}
\end{align*}
Since each recursive call in {\sc DevTree} goes one level deeper in to the tree $T$, the complexity of {\sc DevTree} is bounded by the height of $T$.
\begin{lemma}\label{lem:onescore}
{\sc DevTree} takes $O(h(T))$ time, where $h(T)$ is the height of tree $T$.
\end{lemma}
To compute the $L_1$ costs for all intervals with length $k$, we can dynamically maintain a binary search tree $T$. After the cost for interval $[i,i+k]$ has been computed, we can update the tree to compute interval $[i+1, i+k+1]$ by removing $x_i$ from the tree and adding $x_{i+k+1}$. The detailed process is described in {\sc CostsSizeK} of Algorithm~\ref{alg:lenscore}. Recall that each insert and delete operation on the binary search tree $T$ takes $O(h(T))$ time.
\begin{lemma}\label{lem:lenscore}
{\sc CostsSizeK} takes $O(n\;h(T))$ time.  
\end{lemma}

In addition, the height of the binary search tree $T$ is always $O(\log |T|)$ if we implement $T$ with any balanced binary search tree (e.g. AVL tree, Red-black tree). 
We discuss the overall runtime of {\sc AllCosts} in Section~\ref{sec:runtime}.

\subsection{Efficient algorithm for least cost partition} \label{sec:partintervals}

This section describes the {\sc LeastCostPartition} algorithm.  A pseudocode description is shown in Algorithm~\ref{alg:1dhistogram}.

\begin{algorithm}
\caption{Algorithm for choosing least cost partition}
\label{alg:1dhistogram}
    \begin{algorithmic}
    \Procedure{LeastCostPartition}{$\allbuckets, cost$} 
    \State $\B_j\leftarrow \emptyset$, \;$j=1, 2,\ldots, n$
    \State $c_j\leftarrow +\infty$, \;$j=1, 2,\ldots, n$
    \For {$j=1, \ldots, n$}
    	\ForAll {$b=[i, j]\in\allbuckets$}
		\If {$\B_{i-1}\neq\emptyset$ and $c_{i-1}+cost[b]<c_j$}
			\State $c_j \gets c_{i-1}+cost[b]$
			\State $\B_j \gets \{b\}\cup \B_{i-1}$
		\EndIf
	\EndFor
    \EndFor
    \State \Return $\B_n$
    \EndProcedure
    \end{algorithmic}
\end{algorithm}

Given a set of buckets $\allbuckets$, as well as the cost of each bucket $b\in\allbuckets$, an optimal partition $\B$ can be constructed using dynamic programming. The key idea of this algorithm is same with the algorithm of computing the v-optimal histogram in \cite{Jagadish:1998:OHQ:645924.671191}. We incrementally build an optimal partition that covers $[1, j]$ for $j = 1,\dots, n$.  Let $\B_j$ denote an optimal partition over $[x_1, x_j]$.  The last bucket of $\B_j$ is $[i^*, j]$ for some $1 \leq i^* \leq j$.  Then it must be that the remaining buckets of $\B_j$ form an optimal partition for $[1, \ldots, i^* -1]$.  Otherwise, we can come up with a partition on $[1, j]$ with lower cost by combining the optimal partition on $[1, i^* -1]$ and bucket $[i^*, j]$.  

Thus, the optimal partition on $x_1, \ldots, x_j$ can be constructed using $\B_1, \ldots, \B_{j-1}$:
\begin{align*}
\B_j &= \B_{i^*-1} \cup \{[i^*, j]\}, \text{where}\\
i^* &= \operatornamewithlimits{argmin}_i\{cost(\B_{i-1})+cost([i,j])\;|\;[i,j]\in\allbuckets\}.
\end{align*}

In addition, since each bucket $b$ is visited only once by the inner loop of Algorithm~\ref{alg:1dhistogram}, the running time of the algorithm is linear to $n$ and the number of buckets in $\allbuckets$.
\begin{lemma}\label{lem:1dhistogram}
The running time of Algorithm~\ref{alg:1dhistogram} is $O(|\allbuckets|+n)$.
\end{lemma}

\subsection{Time complexity of Algorithm~\ref{alg:privpartintervals}} \label{sec:runtime}

Having described and analyzed the subroutines {\sc CostsSizeK} and {\sc LeastCostPartition}, we can now discuss the runtime of the entire algorithm.  First, we consider the case when all possible intervals are calculated.  Thus, $|\allbuckets| = O(n^2)$ and {\sc AllCosts} requires $O(n^2 \log n)$ time.  The runtime of Algorithm~\ref{alg:privpartintervals} is dominated by time spent computing {\sc AllCosts}.

\begin{theorem}\label{thm:runtime}
    Algorithm~\ref{alg:privpartintervals} takes $O(n^2\log n)$ time.
\end{theorem}

We can achieve a lower runtime by considering partitions that consist only of intervals whose lengths are powers of 2.  In other words, $\allbuckets = \set{ [i,j] \myvert j-i+1 \text{ is a power of 2}}$ and $|\allbuckets| = O(n \log n)$.  This restriction requires making only a minor modification to the algorithm: the for loop in {\sc AllCosts} now ranges over $k = 1, 2, 4, \ldots, 2^{\lfloor \log n \rfloor }$. 
Constructing such a partition takes much less time.
\begin{theorem}
When Algorithm~\ref{alg:privpartintervals} is restricted to only compute costs for intervals whose length is a power of 2, the algorithm takes $O(n\log^2 n)$ time.
\end{theorem}


\section{Running Time}

The following are running times, for $\epsilon=0.1$, for each of the algorithms considered in Sec \ref{sec:experiments}.  All algorithms are quite efficient, with the exception of MWEM on the harder datasets.  \DW runs in a few seconds on all datasets studied.  Its increases with the complexity of the data set, as expected, but only by a few seconds.  For easier datasets, the \DW running time almost matches that of standard mechanisms like Identity.

\begin{table}[h]
\centering
\small
\begin{tabular}{|c|c|c|c|c|c|c|c|c|c|c|}
\hline
Dataset & DAWA & Partition+Laplace & Partition+Geometric & Partition+Uniform & Identity & Privelet & EFPA & P-HP & S. First & MWEM\\
\hline Nettrace & 1.33 & 1.17 & 1.17 & 1.17 & 1.15 & 1.21 & 1.21 & 1.73 & 35.45 & 11.60\\
\hline Adult & 1.31 & 1.17 & 1.17 & 1.17 & 1.14 & 1.21 & 1.21 & 1.69 & 44.56 & 11.55\\
\hline Medical Cost &  1.36 & 1.17 & 1.17 & 1.17 & 1.14 & 1.21 & 1.21 & 1.76 & 41.89 & 11.65\\
\hline Search Logs & 1.91 & 1.18 & 1.21 & 1.22 & 1.14 & 1.20 & 1.21 & 1.90 & 45.40 & 151.1\\
\hline Income & 3.05 & 1.22 & 1.28 & 1.28 & 1.16 & 1.22 & 1.22 & 1.76 & 38.95 & 2982 \\
\hline Patent & 5.23 & 1.22 & 1.36 & 1.36 & 1.14 & 1.20 & 1.20 & 1.89 & 48.03 & 1568 \\
\hline HepPh & 3.80 & 1.22 & 1.28 & 1.28 & 1.14 & 1.21 & 1.21 & 1.90 & 52.92 & 39.98 \\
\hline
\end{tabular}
\caption{Running time (in seconds) for $\epsilon=0.1$}
\end{table}

\section{Proofs}
This section covers all omitted proofs in the paper.
\subsection{Proofs in Section~\ref{sec:partition}}
\thmutilpartition*
\begin{proof}
Let $S_t$ be the set of histograms such that if $H \in S_t$, then $cost(H) \leq OPT + t$.  Let $\bar{S}_t$ be the complement of $S_t$.

Let $Z_B$ denote the Laplace random variable with scale $\lambda = 2 \sens \bcost / \epsilon$ that is added to the cost of bucket $B$ for each $B \in \allbuckets$.  

Consider the event that $|Z_B| < \frac{t}{2n}$ for all $B \in \allbuckets$.  If this happens, then the algorithm will return a solution in $S_t$.  There are at most $n$ buckets in any histogram, so after noise is added to the individual bucket costs, the cost of the histogram can change by an amount which must be less than $n \frac{t}{2n} = \frac{t}{2}$.  Thus, any histogram in $\bar{S}_t$ will still have a higher cost than the optimal histogram and thus not be selected.

Let $\alg$ denote the algorithm.
\begin{align*}
    &P \left( \alg(\db) \in S_t \right) \\
    &\geq P \left( |Z_B| < \frac{t}{2n} \text{ for all } B \in \allbuckets \right) \\
    &= 1 - P \left( \exists B \in \allbuckets \text{ such that } |Z_B| \geq \frac{t}{2n} \right) \\
    &\geq 1 - |\allbuckets| P \left( |Z_B| \geq \frac{t}{2n} \right) & \text{(union bound)} \\
    &= 1 - |\allbuckets| \exp\left( - \frac{t}{\lambda 2n} \right) & |Z_B| \sim \text{Exp}(1/\lambda) \\
    &\geq 1 - \delta
\end{align*}
when $t \geq \lambda \; 2n \log(|\allbuckets|/\delta) = \frac{4 \sens c \; n \log(|\allbuckets|/\delta)}{\epsilon}$.
\end{proof}

\subsection{Proofs in Section~\ref{sec:stats}}\label{sec:app:proof4}
\propqueryconvert*
\begin{proof}
Given a vector of statistics $\s = s_1 \ldots s_k$ for the corresponding buckets $B$, 
the estimate for the data vector $\x$ is constructed by uniform expansion $\estx = \unexp(\B, \s)$:
$$\hat{x}_i = \frac{s_j}{|b_j|},\quad i\in b_j.$$
Given a query $q=(q_1, \ldots, q_n)$ on $\x$, an estimated answer to $q$, $q(\hat\x)$, is computed as
\begin{equation}\label{eqn:queryconvert}
q(\hat{\x}) = \sum_{j=1}^k\sum_{i\in b_j} q_i\frac{s_j}{|b_j|} = \sum_{j=1}^k\left(\sum_{i\in b_j}\frac{q_i}{|b_j|}\right)s_j.
\end{equation}
\end{proof}

\thmgreedycomplex*
\begin{proof}
Recall $[i,j]$ is used to denote the interval that corresponding to $q$. For each iteration of the loop in Algorithm~\ref{alg:greedyhier}, we would like to prove that it takes $O(m(j-i+1) + (j-i+1)^2)$ to finish the preprocessing and $O(1)$ to compute Eqn.~(\ref{eqn:decaypartialmse}).

Notice that all queries in $Y_{q_1}$ and $Y_{q_2}$ is contained in range $[i,j]$, and query $q$ is the sum of all entries in range $[i, j]$. Since $Y_q$ is a union of $Y_{q_1}$, $Y_{q_2}$, and the query $q$ itself, the matrix $\Y_q$ must be in the following form:
\begin{align*}
\Y_q =&\left[\begin{array}{ccccccccc} 
0 & \ldots & 0 & 1 & \ldots & 1 & 0 & \ldots & 0 \\
0 & \ldots & 0 & \multicolumn{3}{c}{\Y_{q_1}} & 0 & \ldots & 0 \\
0 & \ldots & 0 & \multicolumn{3}{c}{\Y_{q_2}} & 0 & \ldots & 0 
\end{array}\right].
\end{align*}

To compute $(\Y_q^T\D_q^2\Y_q)^{-1}$ efficiently, the Woodbury formula \cite{woodbury50invert},  is needed
\begin{equation}\label{eqn:woodbury}
(\A + \mathbf{UCV})^{-1} = \A^{-1} - \A^{-1}\mathbf{U}(\mathbf{C}^{-1} + \mathbf{V}\A^{-1}\mathbf{U})^{-1}\mathbf{V}\A^{-1}.
\end{equation}
Since queries in $Y_{q_1}$ and $Y_{q_2}$ don't overlap, for any diagonal matrix $\D$, $\Y_{q_1}^T\D\Y_{q_2} = \Y_{q_2}^T\D\Y_{q_1}=0$. Therefore,
\begin{align*}
&\Y_q^T\D_q^2\Y_q\\
=&\Y_q^T\left[\begin{smallmatrix} \lambda_q^2 & 0 & 0 \\ 0 & (1-\lambda_q)^2\D^2_{q_1} & 0\\ 0 & 0 & (1-\lambda_q)^2\D^2_{q_2} \end{smallmatrix}\right]\Y_q\\
=&\left[\begin{array}{ccc} 
1 &  \multirow{3}{*}{$\Y_{q_1}^T$} & \multirow{3}{*}{$\Y_{q_2}^T$} \\
\vdots \\
1
\end{array}\right]
\left[\begin{smallmatrix} \lambda_q^2 \\ & (1-\lambda_q)^2\D^2_{q_1} \\& & (1-\lambda_q)^2\D^2_{q_2} \end{smallmatrix}\right]
\left[\begin{array}{ccc} 
1 & \ldots & 1 \\
\multicolumn{3}{c}{\Y_{q_1}}\\
\multicolumn{3}{c}{\Y_{q_2}}
\end{array}\right]\\
=&(1-\lambda_q)^2\left[\begin{smallmatrix} \Y_{q_1}^T\D^2_{q_1}\Y_{q_1} & \Y_{q_1}^T\D_{q_1}\Y^2_{q_2}\\ 
\Y_{q_2}^T\D_{q_1}\Y^2_{q_1} & \Y_{q_2}^T\D_{q_1}\Y^2_{q_2}\end{smallmatrix}\right]
+\lambda_q^2\left[\begin{smallmatrix} 1&\ldots &1 \\
\vdots & \ddots & \vdots\\
1&\ldots &1
\end{smallmatrix}\right],\\
=&(1-\lambda_q)^2\left[\begin{smallmatrix} \Y_{q_1}^T\D^2_{q_1}\Y_{q_1} & 0\\ 
0 & \Y_{q_2}^T\D_{q_1}\Y^2_{q_2}\end{smallmatrix}\right]
+\lambda_q^2\left[\begin{smallmatrix} 1&\ldots &1 \\
\vdots & \ddots & \vdots\\
1&\ldots &1
\end{smallmatrix}\right].
\end{align*}
Let
\begin{align*}
\A &= (1-\lambda_q)^2\left[\begin{smallmatrix} \Y_{q_1}^T\D^2_{q_1}\Y_{q_1} & 0\\ 
0 & \Y_{q_2}^T\D_{q_1}\Y^2_{q_2}\end{smallmatrix}\right]\\
\mathbf{V} &= \mathbf{U}^T = [1, \ldots, 1] \\
\mathbf{C} &= \lambda_q^2.
\end{align*}
Then $\Y_q^T\D_q^2\Y_q=\A + \mathbf{UCV}$ can be computed using Eqn.~(\ref{eqn:woodbury}):
\begin{align}
&(\Y_q^T\D_q^2\Y_q)^{-1}\nonumber\\
=&(\A + \mathbf{UCV})^{-1} \nonumber\\
=&\A^{-1} - \A^{-1}\mathbf{U}(\mathbf{C}^{-1} + \mathbf{V}\A^{-1}\mathbf{U})^{-1}\mathbf{V}\A^{-1} \nonumber\\
=&(1-\lambda_q)^{-2}
\left[\begin{smallmatrix} \Y_{q_1}^T\D^2_{q_1}\Y_{q_1} & 0\\ 
0 & \Y_{q_2}^T\D_{q_1}\Y^2_{q_2}\end{smallmatrix}\right]^{-1} \nonumber \\
&-(1-\lambda_q)^{-4}\left[\begin{smallmatrix} \Y_{q_1}^T\D^2_{q_1}\Y_{q_1} & 0\\ 
0 & \Y_{q_2}^T\D_{q_1}\Y^2_{q_2}\end{smallmatrix}\right]^{-1}
\left[\begin{smallmatrix}1 \\ \vdots \\ 1\end{smallmatrix}\right]\left(\lambda_q^{-2}+ [1, \ldots, 1]
\left[\begin{smallmatrix} \Y_{q_1}^T\D^2_{q_1}\Y_{q_1} & 0\\ 
0 & \Y_{q_2}^T\D_{q_1}\Y^2_{q_2}\end{smallmatrix}\right]^{-1}
\left[\begin{smallmatrix}1 \\ \vdots \\ 1\end{smallmatrix}\right]
\right)^{-1} [1, \ldots, 1]\left[\begin{smallmatrix} \Y_{q_1}^T\D^2_{q_1}\Y_{q_1} & 0\\ 
0 & \Y_{q_2}^T\D_{q_1}\Y^2_{q_2}\end{smallmatrix}\right]^{-1}\label{eqn:fastinv}
\end{align}
To simplify Eqn.~\ref{eqn:fastinv}, let
\begin{align}
\v_{q_1} &= (\Y_{q_1}^T\D^2_{q_1}\Y_{q_1})^{-1}\left[\begin{smallmatrix}1 \\ \vdots \\ 1\end{smallmatrix}\right],\label{eqn:compv1}\\
\v_{q_2} &= (\Y_{q_2}^T\D^2_{q_2}\Y_{q_2})^{-1}\left[\begin{smallmatrix}1 \\ \vdots \\ 1\end{smallmatrix}\right],\label{eqn:compv2}\\
m_{q_1} &= [1,\ldots,1](\Y_{q_1}^T\D^2_{q_1}\Y_{q_1})^{-1}\left[\begin{smallmatrix}1 \\ \vdots \\ 1\end{smallmatrix}\right],\label{eqn:compm1}\\
m_{q_2} &= [1,\ldots,1](\Y_{q_2}^T\D^2_{q_2}\Y_{q_2})^{-1}\left[\begin{smallmatrix}1 \\ \vdots \\ 1\end{smallmatrix}\right].\label{eqn:compm2}
\end{align}
and Eqn.~\ref{eqn:fastinv} becomes
\begin{align}
(\mbox{\ref{eqn:fastinv}}) = (1-\lambda_q)^{-2}
\left[\begin{smallmatrix} (\Y_{q_1}^T\D^2_{q_1}\Y_{q_1})^{-1} & 0\\ 
0 & (\Y_{q_2}^T\D_{q_1}\Y^2_{q_2})^{-1}\end{smallmatrix}\right] -(1-\lambda_q)^{-4}(\lambda_q^{-2} + m_{q_1} + m_{q_2})^{-1}\left[\begin{smallmatrix}\v_{q_1}\\ \v_{q_2}\end{smallmatrix}\right]
[\v_{q_1}^T, \v_{q_2}^T]\label{eqn:fastinvsimple}.
\end{align}
Recall $\tWW_q = [\tWW_{q_1},\tWW_{q_2}]$, according to Eqn.~(\ref{eqn:fastinvsimple}), we have
\begin{align}
&tr({\tWW_q}^T\tWW_q(\Y^T_q\D^2_{q}\Y_q)^{-1})\nonumber\\
=&(1-\lambda_q)^{-2}tr\left({\tWW_q}^T\tWW_q\left[\begin{smallmatrix} (\Y_{q_1}^T\D^2_{q_1}\Y_{q_1})^{-1} & 0\\ 
0 & (\Y_{q_2}^T\D_{q_1}\Y^2_{q_2})^{-1}\end{smallmatrix}\right]\right)
- (1-\lambda_q)^{-4}(\lambda_q^{-2} + m_{q_1} + m_{q_2})^{-1}tr({\tWW_q}^T\tWW_q\left[\begin{smallmatrix}\v_{q_1}\\ \v_{q_2}\end{smallmatrix}\right]
[\v_{q_1}^T, \v_{q_2}^T])\nonumber\\
=&(1-\lambda_q)^{-2}tr\left(
\left[\begin{smallmatrix}\tWW_{q_1}^T\\ \tWW_{q_2}^T\end{smallmatrix}\right]
[\tWW_{q_1}, \tWW_{q_2}]\left[\begin{smallmatrix} (\Y_{q_1}^T\D^2_{q_1}\Y_{q_1})^{-1} & 0\\ 
0 & (\Y_{q_2}^T\D_{q_1}\Y^2_{q_2})^{-1}\end{smallmatrix}\right]\right)\nonumber\\
&-(1-\lambda_q)^{-4}(\lambda_q^{-2} + m_{q_1} + m_{q_2})^{-1}tr([\v_{q_1}^T, \v_{q_2}^T]
\left[\begin{smallmatrix}\tWW_{q_1}^T\\ \tWW_{q_2}^T\end{smallmatrix}\right]
[\tWW_{q_1}, \tWW_{q_2}]\left[\begin{smallmatrix}\v_{q_1}\\ \v_{q_2}\end{smallmatrix}\right]
)\nonumber\\
=&(1-\lambda_q)^{-2}tr\left(
\left[\begin{smallmatrix}\tWW_{q_1}^T\tWW_{q_1} & \tWW_{q_1}^T\tWW_{q_2} \\ 
\tWW_{q_2}^T\tWW_{q_1} & \tWW_{q_2}^T\tWW_{q_2}\end{smallmatrix}\right]
\left[\begin{smallmatrix} (\Y_{q_1}^T\D^2_{q_1}\Y_{q_1})^{-1} & 0\\ 
0 & (\Y_{q_2}^T\D_{q_1}\Y^2_{q_2})^{-1}\end{smallmatrix}\right]\right)\nonumber\\
&-(1-\lambda_q)^{-4}(\lambda_q^{-2} + m_{q_1} + m_{q_2})^{-1}
tr((\v_{q_1}^T\tWW_{q_1}^T+ \v_{q_2}^T\tWW_{q_2}^T)
(\tWW_{q_1}\v_{q_1}+ \tWW_{q_2}\v_{q_2}))\nonumber\\
=&(1-\lambda_q)^{-2}tr({\tWW_{q_1}}^T\tWW_{q_1}(\Y^T_{q_1}\D^2_{q_1}\Y_{q_1})^{-1})
+(1-\lambda_q)^{-2}tr({\tWW_{q_2}}^T\tWW_{q_2}(\Y^T_{q_2}\D^2_{q_2}\Y_{q_2})^{-1})\nonumber\\
&-(1-\lambda_q)^{-4}(\lambda_q^{-2} + m_{q_1} + m_{q_2})^{-1}||\tWW_{q_1}\v_{q_1}+ \tWW_{q_2}\v_{q_2}||_2^2
\label{eqn:partialcomp1},
\end{align}
and
\begin{align}
&tr\left(\left[\begin{smallmatrix}{\tWW_{q_1}}^T\tWW_{q_1} & 0\\0 & {\tWW_{q_2}}^T\tWW_{q_2}\end{smallmatrix}\right](\Y^T_q\D^2_{q}\Y_q)^{-1}\right)\nonumber\\
=&(1-\lambda_q)^{-2}tr\left(\left[\begin{smallmatrix}{\tWW_{q_1}}^T\tWW_{q_1} & 0\\0 & {\tWW_{q_2}}^T\tWW_{q_2}\end{smallmatrix}\right]\left[\begin{smallmatrix} (\Y_{q_1}^T\D^2_{q_1}\Y_{q_1})^{-1} & 0\\ 
0 & (\Y_{q_2}^T\D_{q_1}\Y^2_{q_2})^{-1}\end{smallmatrix}\right]\right)\nonumber\\
&-(1-\lambda_q)^{-4}(\lambda_q^{-2} + m_{q_1} + m_{q_2})^{-1}
tr\left(
\left[\begin{smallmatrix}{\tWW_{q_1}}^T\tWW_{q_1} & 0\\ 0 & {\tWW_{q_2}}^T\tWW_{q_2}\end{smallmatrix}\right]
\left[\begin{smallmatrix}\v_{q_1}\\ \v_{q_2}\end{smallmatrix}\right]
[\v_{q_1}^T, \v_{q_2}^T]\right)\nonumber\\
=&(1-\lambda_q)^{-2}tr({\tWW_{q_1}}^T\tWW_{q_1}(\Y^T_{q_1}\D^2_{q_1}\Y_{q_1})^{-1})
+(1-\lambda_q)^{-2}tr({\tWW_{q_2}}^T\tWW_{q_2}(\Y^T_{q_2}\D^2_{q_2}\Y_{q_2})^{-1})\nonumber\\
&-(1-\lambda_q)^{-4}(\lambda_q^{-2} + m_{q_1} + m_{q_2})^{-1}(||\tWW_{q_1}\v_{q_1}||_2^2+ ||\tWW_{q_2}\v_{q_2}||_2^2)
\label{eqn:partialcomp2}.
\end{align}
Eqn~(\ref{eqn:decaypartialmse}) can be computed by putting Eqn.~(\ref{eqn:partialcomp1}) and (\ref{eqn:partialcomp2}) together:
\begin{align}
&tr\left(\left(t^{-\frac{l}{2}}{\tWW_q}^T\tWW_q + (1-t^{-\frac{l}{2}})\left[\begin{smallmatrix}{\tWW_{q_1}}^T\tWW_{q_1} & 0\\ 0 & {\tWW_{q_2}}^T\tWW_{q_2}\end{smallmatrix}\right]\right)(\Y^T_q\D^2_{q}\Y_q)^{-1}\right)\nonumber\\
=&t^{-\frac{l}{2}}tr({\tWW_q}^T\tWW_q(\Y^T_q\D^2_{q}\Y_q)^{-1})
+ (1-t^{-\frac{l}{2}})tr\left(\left[\begin{smallmatrix}{\tWW_{q_1}}^T\tWW_{q_1} & 0\\ 0 & {\tWW_{q_2}}^T\tWW_{q_2}\end{smallmatrix}\right](\Y^T_q\D^2_{q}\Y_q)^{-1}\right)\nonumber\\
=&(1-\lambda_q)^{-2}tr({\tWW_{q_1}}^T\tWW_{q_1}(\Y^T_{q_1}\D^2_{q_1}\Y_{q_1})^{-1})
+(1-\lambda_q)^{-2}tr({\tWW_{q_2}}^T\tWW_{q_2}(\Y^T_{q_2}\D^2_{q_2}\Y_{q_2})^{-1})\nonumber\\
&-(1-\lambda_q)^{-4}(\lambda_q^{-2} + m_{q_1} + m_{q_2})^{-1}
(t^{\frac{l}{2}}||\tWW_{q_1}\v_{q_1}+ \tWW_{q_2}\v_{q_2}||_2^2 + (1-t^{\frac{l}{2}}(||\tWW_{q_1}\v_{q_1}||_2^2+ ||\tWW_{q_2}\v_{q_2}||_2^2))\label{eqn:fastdecaypartialmse}
\end{align}
Observe Eqn.~(\ref{eqn:fastdecaypartialmse}), if we can precompute the following quantities:
\begin{align*}
tr({\tWW_{q_1}}^T\tWW_{q_1}(\Y^T_{q_1}\D^2_{q_1}\Y_{q_1})^{-1}), \:
tr({\tWW_{q_2}}^T\tWW_{q_2}(\Y^T_{q_2}\D^2_{q_2}\Y_{q_2})^{-1}), \:
||\tWW_{q_1}\v_{q_1}+ \tWW_{q_2}\v_{q_2}||_2^2, \:
||\tWW_{q_1}\v_{q_1}||_2^2+ ||\tWW_{q_2}\v_{q_2}||_2^2, \:
m_{q_1}, \: m_{q_2},
\end{align*}
for any given $\lambda_q$, Eqn.~(\ref{eqn:decaypartialmse}) can be computed in $O(1)$ time using Eqn.~(\ref{eqn:fastdecaypartialmse}). 

Since $(\Y^T_{q_1}\D^2_{q_1}\Y_{q_1})^{-1})$, $(\Y^T_{q_1}\D^2_{q_1}\Y_{q_1})^{-1})$, $tr({\tWW_{q_1}}^T\tWW_{q_1}(\Y^T_{q_1}\D^2_{q_1}\Y_{q_1})^{-1})$ and $tr({\tWW_{q_2}}^T\tWW_{q_2}(\Y^T_{q_2}\D^2_{q_2}\Y_{q_2})^{-1})$ have already computed when finding the scaling of $q_1$ and $q_2$, we can get them without extra computation.  Further, $\v_1$ and $\v_2$ can be computed using Eqn.~(\ref{eqn:compv1}) and (\ref{eqn:compv2}), respectively, each of which takes $O((j-i+1)^2)$ time; $\tWW_{q_1}\v_{q_1}$ and $\tWW_{q_2}\v_{q_2}$ can then be computed in $O(m(j-i+1))$ time; $m_1$ and $m_2$ can be computed using Eqn.~(\ref{eqn:compm1}) and (\ref{eqn:compm2}), respectively, each of which takes $O((j-i+1)^2)$ time. 

Above all, by keeping $(\Y^T_{q_1}\D^2_{q_1}\Y_{q_1})^{-1})$, $(\Y^T_{q_1}\D^2_{q_1}\Y_{q_1})^{-1})$, $tr({\tWW_{q_1}}^T\tWW_{q_1}(\Y^T_{q_1}\D^2_{q_1}\Y_{q_1})^{-1})$ and $tr({\tWW_{q_2}}^T\tWW_{q_2}(\Y^T_{q_2}\D^2_{q_2}\Y_{q_2})^{-1})$ from the previous step, we show that it only takes $O(m(j-i+1)+(j-i+1)^2)$ time to do the precomputation in each iteration of the loop of Algorithm~\ref{alg:greedyhier} and the rest of computations can be done in $O(1)$ time. Update scalings takes $O(j-i+1)$ time, which is dominated by the pre-computation time.

The intermediate results in Algorithm~\ref{alg:greedyhier} can also accelerate the least square process. Since we have already computed $(\Y^T\D_Y^2\Y)^{-1}$ in the loop of Algorithm~\ref{alg:greedyhier}, applying the ordinary least square method (the last step of Algorithm~\ref{alg:greedyhier}) only takes $O(k^2)$ time instead of $O(k^3)$ time in general cases.  Summing the costs together proves the theorem.
\end{proof}

\end{document}